%% file: paper_revision.tex
\documentclass{article}
\usepackage[utf8]{inputenc}
\usepackage[english]{babel}
\usepackage[colorinlistoftodos]{todonotes}
\usepackage{amsmath,amsfonts,amssymb,hyperref,color,graphicx,siunitx,fancyhdr,lastpage,enumerate, commath,amsthm, algorithm, algpseudocode,mathrsfs }
\usepackage{url}
\usepackage[numbers]{natbib}

\input{macros.tex}

\title{Weak Poincar{\'e} inequality comparisons for Ideal and Hybrid Slice Sampling}
\author{Sam Power \and Daniel Rudolf \and Bj\"orn Sprungk \and Andi Q. Wang}
\date{August 22, 2025}

\begin{document}

\maketitle

\begin{abstract}
    Using the framework of weak Poincar{\'e} inequalities, we provide a general comparison between Hybrid and Ideal Slice Sampling 
    in terms of their corresponding Dirichlet forms.
    In particular, under suitable assumptions Hybrid Slice Sampling  inherits fast convergence from Ideal Slice Sampling and conversely. 
    We apply our results to analyse the convergence of the Independent Metropolis-Hastings, stepping-out and shrinkage, as well as Hit-and-Run within slice sampling algorithms.
\end{abstract}

\section{Introduction}

Markov chain Monte Carlo (MCMC) methods form the backbone for computational Bayesian inference as a tool for sampling from a (posterior) probability distribution $\pi$ on a measurable state space $(\mathsf{G}, \mathcal{G})$.  Given a suitable $\sigma$-finite dominating reference measure $\nu$ on $(\mathsf{G}, \mathcal{G})$, we represent the distribution $\pi$ as  
\begin{equation}
    \pi(\dif  x) = c^{-1} \,\varpi(x) \,\nu({\rm d}x),
    \qquad
    c = \int_{\mathsf G} \varpi(x)  \nu({\rm d}x),
    \label{eq: repres_pi}
\end{equation}
where the not-necessarily-normalized density function $\varpi \colon \mathsf{G} \to (0, \infty)$ satisfies $\varpi \propto \frac{\dif  \pi}{\dif  \nu}$ with $c\in(0,\infty)$.  
Methodological milestones in the development of MCMC include the Metropolis--Hastings (MH) algorithm as well as Gibbs and Slice Sampling. In the MH algorithm, a proposal is drawn and then with a $\pi$-dependent acceptance probability, either accepted or rejected.  In order to achieve efficient state space exploration of the corresponding Markov chain, step-sizes must be tuned to stabilize the acceptance probabilities away from $0$. More broadly, such approaches require careful design and tuning of proposal kernels, which can significantly influence the overall performance. By contrast, ideal Gibbs and Slice Sampling approaches offer the promise of being `tuning-free' in principle. However, in many cases, implementation of these ideal algorithms can be challenging.
In this work, we focus on the \textit{Slice Sampling} paradigm, popularised in \cite{Neal2003Slice}, with ideas previously developed in papers such as \cite{Edwards1988,Mira2002}, and associated `hybrid' approaches. This hybrid methodology, coined in \cite{Roberts1997geometric,Latuszynski2014convergence}, allows users to balance the goals of efficiency, the degree of necessary tuning and ease of implementation. 

We begin by presenting slice sampling in its ideal form, and then highlight practical advantages regarding the implementation of hybrid schemes. For simplicity, consider the case of $\mathsf{G} = \mathbb{R}^d$ with $\nu$ being the Lebesgue measure on $\mathbb R^d$. In this case, Ideal Slice Sampling involves alternatively drawing an auxiliary `slice' variable $t$, followed by realizing a sample from the normalized restriction of $\nu$ to a corresponding superlevel set (or \textit{slice}),
\begin{equation*}
    \sfG(t) := \left\{ x \in \mathsf G: \varpi(x)> t \right\}.
\end{equation*}
For illustrative purposes a single transition\footnote{In the algorithmic transition mechanisms and elsewhere all appearing random variables are assumed to be defined on a sufficiently rich common probability space. Moreover, whenever it is well-defined, for a set $M$ denote by $\mathrm{Unif}(M)$ the corresponding uniform distribution.} of Ideal Slice Sampling 
from $x\in \sfG$ to $y\in \sfG$ is presented in Algorithm~\ref{alg:SSS}. 
\begin{algorithm}
    \caption{Ideal (uniform) Slice Sampling}
    \begin{algorithmic}[1]
        \State Given target $\pi$ on $\sfG=\R^d$, current location $x\in \sfG$; 
        \State \textbf{Sample} 
        $T\sim \mathrm{Unif}((0,\varpi(x)))$, call the result $t$;
        \State \textbf{Sample} 
        $Y\sim \mathrm{Unif}(\sfG(t))$, call the result $y$;      
        \State \textbf{Return} $y\in \sfG$.
    \end{algorithmic}
    \label{alg:SSS}
\end{algorithm}

\noindent
It is well-known that the corresponding Markov kernel is $\pi$-reversible and therefore has the desirable property of $\pi$-invariance. Provided that such an ideal sampler can be implemented, a number of empirical and theoretical results demonstrate the ability of the slice sampler to explore complex state spaces efficiently; see recent quantitative results on convergence rates in \cite{Natarovskii2021quantitative,Schar2023revisited,RuSch24} for the state of the art. Moreover, let us emphasize that Algorithm~\ref{alg:SSS} is tuning-free: no user-chosen parameters appear that influence the performance. However, practical implementation of Algorithm~\ref{alg:SSS} can be extremely challenging, as the task of sampling from the superlevel sets $\mathsf{G}(t)$ is in general intractable.

A reasonable workaround for a given $\mathsf{G}(t)$ is to use an inner transition mechanism that \textit{approximates} the challenging task of sampling from $\mathsf{G}(t)$. This is commonly referred to as \textit{Hybrid Slice Sampling} \cite{Roberts1997geometric, Latuszynski2014convergence}: instead of sampling exactly from $\nu$ conditioned on the slice $\sfG(t)$, a hybrid method performs a single transition of a Markov kernel $H_t$, which is assumed to be invariant with respect to (w.r.t.) said distribution. 
We refer to the collection of $H_t$ as \emph{on-slice kernels} and note that the feasible design of those is generally nontrivial, but a number of sensible default choices exist \cite{Neal2003Slice, Murray2010elliptical, ScharEtAl2023Gibbsian}. Moreover, under very mild conditions, this approach also gives rise to a $\pi$-invariant Markov chain. Algorithm~\ref{alg:HSS} offers a description of the  corresponding transition mechanism from $x\in \sfG$ to $y\in \sfG$ based on the family of Markov kernels $(H_t)_{t\in(0,\Vert \varpi \Vert_\infty)}$.
\begin{algorithm}
    \caption{Hybrid Slice Sampling}
    \begin{algorithmic}[1]
        \State Given target $\pi$, family of Markov kernels 
        $(H_t)_{t\in(0,\Vert \varpi \Vert_\infty)}$,
        with each $H_t$ being invariant w.r.t. $\mathrm{Unif}(\sfG(t))$, and current location $x$;
        \State \textbf{Sample} $T \sim \mathrm{Unif} \left( (0, \varpi(x) ) \right)$, call the result $t$;
        \State \textbf{Sample} $Y \sim H_t \left( x, \cdot \right)$, call the result $y$;
        \State \textbf{Return} $y\in \sfG$.
    \end{algorithmic}\label{alg:HSS}
\end{algorithm}
We investigate how fast the distribution of the $n$th iterate of Hybrid Slice Sampling converges, \textit{relative} to 
their idealised versions. 
Building on Markov chain comparison techniques developed recently in \cite{Andrieu2021comparison}, we show that (in a sense to be made precise) Hybrid Slice Samplers are at best as good as their ideal counterparts, and furthermore, that under suitable assumptions this statement can be complemented by verifying that they are `almost' as good as their ideal versions. The proper statement of the latter descriptions is provided in Theorem~\ref{thm:SS_comp}.
Combined with results on the convergence of Ideal Slice Samplers, such as \cite{Natarovskii2021quantitative} and the recent developments in \cite{Schar2023revisited,RuSch24}, one can then deduce bounds on the  convergence rate of Hybrid Slice Samplers in various settings.

\emph{Outline:} The comparison framework based on weak Poincar\'e inequalities is introduced in Section~\ref{subsec:general_comp}.  It is then applied to Hybrid Slice Sampling in Section \ref{subsec:comp_SS}.  In Section~\ref{sec:examples} we show how the novel abstract comparison result can be employed to study a number of classical Metropolis algorithms (upon observing that they are, in fact, instances of Hybrid Slice Sampling) as well as stepping-out and shrinkage on bimodal targets.
Finally, we provide a detailed analysis of \textit{Hit-and-Run within Slice Sampling}.

\subsection{Related work}
Convergence of Ideal Slice Sampling algorithms has been studied in the recent works of \cite{Natarovskii2021quantitative, Schar2023revisited,RuSch24}. Moreover, comparisons of slice sampling with hybrid variants and other algorithms have been initiated and investigated in \cite{Latuszynski2014convergence, Rudolf2018comparison} as well as in the recent work of \cite{Qin2023}. The previous results on convergence of Hybrid Slice Samplers, \cite[Theorem~8]{Latuszynski2014convergence} and the improvement \cite[Proposition~28]{Qin2023}, give sufficient conditions for geometric convergence.

In this work we give a substantial extension in Theorem~\ref{thm:SS_comp}. We utilize the weak Poincar\'e inequality comparison framework introduced in \cite{Andrieu2021comparison}, where it was used to study pseudo-marginal MCMC methods, and further discussed in \cite{Andrieu2022poincare,Andrieu2023wpi}. 
We remark that while we utilize the same high-level comparison technique as \cite{Andrieu2021comparison}, the nature of the comparison is notably different: in the case of pseudo-marginal MCMC, it was the \textit{acceptance probability} that was perturbed in a `local' manner, whereas in the (Ideal vs. Hybrid) Slice Sampling setting, we are perturbing the 
\textit{on-slice transition mechanism}, thus demonstrating that the comparison technique is able to flexibly handle a range of perturbations.

As already mentioned, our Theorem~\ref{thm:SS_comp} significantly extends previous results of \cite{Latuszynski2014convergence,Qin2023}, both by weakening the conditions required to obtain conclusions, and enabling the quantitative treatment of such Hybrid Slice Sampling with no spectral gap and slower-than-exponential convergence behaviour, i.e. in the \textit{subgeometric} scenario. The geometric case is also recovered as a special case; see Corollary~\ref{cor:Ht_gap_comp}.

\section{Preliminaries}\label{subsec:general_comp}

We provide some general notation and preliminaries. 
Let $\left( \sfG, \mathcal{G}, \mu \right)$
be a generic probability space
and for $f \colon \mathsf{G} \to \R$ let $\mu(f) := \int_\sfG f(x) \mu (\dif x)$ be the mean of $f$ w.r.t. $\mu$ whenever it exists. We let $\ELL(\mu)$ denote the space of square-integrable functions that is determined by the set of functions $f\colon \sfG\to\R$ that satisfy $\| f \|^2_\mu := \mu(f^2)<\infty$. Moreover, the corresponding inner product is given by $\langle f, g \rangle_\mu := \mu(f \cdot g)$ for $f, g \in \ELL(\mu)$. 
For $f\in \ELL(\mu)$ we denote the variance of $f$ w.r.t. $\mu$ as
\begin{equation*}
    \Var_\mu(f) := \mu(f^2) -\mu(f)^2.
\end{equation*}
Obviously, if $f\in \ELL_0(\mu) := \{ f \in \ELL(\mu): \mu(f) = 0 \}$, then $\Var_\mu(f) = \|f\|^2_\mu$. For $f\colon \sfG\to\R$ we also define the oscillation seminorm as
\begin{equation*}
       \|f\|_{\mathrm{osc},\mu}:= \ess_\mu \sup f - \ess_\mu \inf f.
\end{equation*}
Let $P\colon \sfG\times \mathcal{G} \to [0,1]$ be a generic Markov kernel that is \textit{$\mu$-reversible}, i.e. $\mu \otimes P (A \times B) = \mu \otimes P(B \times A)$ for all $A, B \in \mathcal{G}$, where $\mu \otimes P \left( \dif  x, \dif  y \right) =  P \left( x, \dif  y \right) \mu \left( \dif  x \right)$.
The $\mu$-reversibility of $P$ implies $\mu$-invariance, that is, $\mu P(A) := \int_G P(x,A)\,\mu(\dif x)$ coincides with $\mu(A)$ for any $A\in\mathcal{G}$.
Note that the Markov kernel $P$ induces a self-adjoint linear operator $\mathrm{P} \colon \ELL(\mu)\to  \ELL(\mu)$ by $f\mapsto \mathrm{P}f$, where $\mathrm{P}f(x):= \int_{\sfG} f(y) P(x,\dif y)$ for $x\in\sfG$. Moreover, $P$ is said to be \textit{positive } if $\langle \mathrm{P}f, f\rangle_\mu \geq 0$ for all $f\in \ELL(\mu)$.
In this context the \textit{Dirichlet form} of $P$ is given by 
\begin{equation*}
    \mathcal E_\mu(P,f) := \langle (\Id - \mathrm{P}) f, f \rangle_\mu, \quad f\in \ELL(\mu),
\end{equation*}
where $\Id$ denotes the identity operator on $\ELL(\mu)$. 
Exploiting the $\mu$-invariance (implied by the $\mu$-reversibility) yields the representation
\begin{equation}
    \calE_\mu (P,f) = \frac{1}{2} \int_{\sfG} \int_{\sfG}   \left[ f(x) - f(y) \right]^2  P \left( x, \dif y \right)  \mu (\dif x).
    \label{eq:dirichlet_repn}
\end{equation}
Given two $\mu$-reversible Markov kernels kernels $P_1, P_2$, we say that $P_1$ \textit{dominates} $P_2$ if the Dirichlet forms satisfy for all $f\in \ELL_0(\mu)$ that
$
    \calE_\mu (P_2,f) \leq \calE_\mu (P_1,f).
$%
 We now detail concepts concerning weak Poincar\'e inequalities.

\subsection{Weak Poincar\'e inequalities and consequences}

We follow the approach to weak Poincar\'e inequalities developed in \cite{Andrieu2021comparison, Andrieu2022poincare} in the context of (discrete-time) Markov chains  imposing additional regularity properties (such as having a reversible and positive Markov kernel), which will be sufficient for our purposes. 

Let $\left( \sfG, \mathcal{G}, \mu \right)$ be a probability space. We provide the definition 
and consequences of the property of a weak Poincar\'e inequality.

\begin{definition}\label{def:WPI}
    For a function $\beta \colon (0,\infty) \to [0,\infty)$ we say that a Markov kernel $P\colon \sfG\times \mathcal{G} \to [0,1]$ satisfies a weak Poincar{\'e} inequality (WPI) w.r.t. $\beta$ 
    if $\beta$ is decreasing with $\lim_{s\to\infty} \beta(s)= 0$ and
    $P$ is $\mu$-reversible, positive and  
    for all $s>0$, $f\in \ELL_0(\mu)$ we have
        \begin{equation}
         \|f\|^2_\mu \leq s\cdot \calE_\mu (P,f) + \beta(s) \cdot \|f\|_{\osc,\mu}^2.
         \label{eq:WPI}
    \end{equation}
    If there exists a function $\beta$, such that $P$ satisfies a WPI w.r.t. $\beta$ we also just say $P$ satisfies a WPI.
\end{definition}

\begin{remark}
    If a $\mu$-invariant Markov kernel $P$ is $\mu$-irreducible, then there exists a function $\beta$ satisfying the properties of Definition~\ref{def:WPI}, see \cite[Section~4.2.3]{Andrieu2022poincare}. 
    Moreover, given a Markov kernel that satisfies a WPI with $\beta$, by Popoviciu's inequality $\|f\|^2_\mu \le\|f\|_{\osc,\mu}^2/4$ for any $f\in \ELL_0(\mu)$, it also satisfies a WPI w.r.t. $s\mapsto\min\{1/4,\beta(s)\}$. Note that this new function is uniformly bounded by $1/4$.
    \label{rmk:beta}
\end{remark}
To highlight the utility of establishing a WPI for a Markov kernel, we reformulate a result that is stated in \cite[Theorem~7, Example~2]{Andrieu2022poincare}, see also \cite[Theorem~21, Remark 2]{Andrieu2021comparison}.
\begin{theorem} \label{thm:conv}
Assume that a Markov kernel $P\colon \sfG\times\mathcal{G}\to [0,1]$ satisfies a WPI w.r.t. $\beta$. Then, there exists a strictly decreasing function $\alpha_\beta \colon (0,\infty) \to (0,1)$ with $\lim_{n\to\infty} \alpha_\beta(n) = 0$ such that for all $n\in\mathbb{N}$ and  $f\in \ELL_0(\mu)$ with $\Vert f \Vert_{\osc,\mu}<\infty$  we have
\begin{equation}    \label{eq: to_zero}
        \|\mathrm{P}^n f\|^2_\mu \leq \alpha_\beta(n) \cdot \|f\|_{\osc,\mu}^2.
\end{equation}
\end{theorem}
For the convenience of the reader
we provide a consequence that motivates the interest in WPIs in the context of the convergence of Markov chains that readily follows by  virtue of \cite[Remark~2, Supplement]{Andrieu2021comparison}.
\begin{corollary}
\label{cor: MC_implication}
    Let $(\Omega,\mathcal{A},\mathbb{P})$ be a probability space and $(X_n)_{n\in\mathbb{N}_0}$ be a Markov chain with Markov kernel $P\colon \sfG\times \mathcal{G}\to [0,1]$ and initial distribution $\mu_0$, that is, for all $i\in\mathbb{N}_0$ the mappings $X_i\colon (\Omega,\mathcal{A}) \to (\sfG,\mathcal{G})$ are measurable, satisfy
    \begin{align*}
        \mathbb{P}(X_{i+1}\in A\mid X_1,\dots,X_{i}) 
        = \mathbb{P}(X_{i+1}\in A\mid X_{i}) = P(X_{i},A)
    \end{align*}
    for all $A\in\mathcal{G}$ almost surely and $\mathbb{P}^{X_0}=\mu_0$, where $\mathbb{P}^{X_j}$ denotes the distribution of $X_j$ for any $j\in\mathbb{N}_0$. Moreover, assume that for $P$ a WPI w.r.t. $\beta$ holds, as well as, that $\mu_0$ admits a density $\frac{ \dif \mu_0}{\dif \mu}$ w.r.t. $\mu$. Then, there is a strictly decreasing function $\alpha_\beta \colon (0,\infty) \to (0,1)$ with $\lim_{n\to\infty} \alpha_\beta(n) = 0$ such that
    \[
    \sup_{A\in\mathcal{G}} \left\vert \mathbb{P}^{X_n}(A) - \mu(A) \right\vert 
    \leq \left\Vert \frac{ \dif \mu_0}{\dif \mu}-1 \right\Vert_{\osc,\mu}^{1/2} \cdot \alpha_\beta(n)^{1/2}, \qquad n\in\mathbb{N}.
    \]
\end{corollary}
For further consequences regarding corresponding Markov chains we refer to \cite[Section~2.4]{Andrieu2022poincare}. 
Let us now comment on different aspects of $\alpha_\beta$.
\begin{remark}
\label{rem: alpha_n}
    In \cite{Andrieu2021comparison,Andrieu2022poincare}
    the function $\alpha_\beta$ is determined as the inverse of the strictly decreasing convex function $F_\beta\colon (0,1] \to (0,\infty)$ given by $F_\beta(x):=\int_x^1 (K_\beta^*(v))^{-1} {\dif v}$. Here $K_\beta^*\colon [0,\infty)\to [0,\infty]$ is defined as $K_\beta^*(v)=\sup_{u\geq0}[uv-K_
    \beta(u)]$, the convex conjugate of $K_\beta\colon[0,\infty)\to [0,\infty)$, which itself is given by
\begin{align*}
    K_\beta(u):=  \begin{cases}
            u \cdot \beta(1/u) & u>0,\\
            0   & u=0.
        \end{cases}
\end{align*}
The above abstract construction is somewhat opaque; key examples are given in \cite[Section~2.1.2]{Andrieu2021comparison}: for constants $c_0,c_1,c_2,c_3 \in (0,\infty)$ there exist further constants $c_4,c_5\in (0,\infty)$ such that
   \begin{itemize}
        \item $\beta(s)=c_0 s^{-c_1}$ yields $\alpha_\beta(n) \leq c_0(1+c_1)^{1+c_1} n^{-c_1}$,
        \item $\beta(s) = c_1 \exp \left( - c_2 s^{c_3} \right)$ yields $\alpha_\beta(n) \leq c_4 \exp\big( - c_5 \cdot n^\frac{c_3}{1 + c_3} \big)$, and
        \item $\beta(s) = c_1 / \log (\max \left\{ c_2, s \right\}) ^{c_3}$ yields $\alpha_\beta(n) \leq c_4 / \log (\max \left\{ 2, n \right\}) ^{c_3}$.
    \end{itemize}  
\end{remark}

We finish this section by emphasizing the relation of a WPI to a (strong) Poincar\'e inequality that yields a lower bound of the (absolute) $L^2(\mu)$-spectral gap of $P$.

\begin{remark}
\label{rem: SPI}
    Let $P$ be a Markov kernel that satisfies a WPI w.r.t. $\beta$. Assume that for $\beta$ there exists a $\gamma>0$, such that $\beta(s) = 1/4\cdot\mathbf{1}_{(0,\gamma^{-1})}(s)$.
    Then, \eqref{eq:WPI} is equivalent to $\gamma \,\|f\|^2_\mu \leq \mathcal E_\mu(P,f)$ for any $f\in \ELL_0(\mu)$, known as the (strong) Poincar\'e inequality. 
 The latter is also equivalent to 
 \begin{equation}
 \label{eq: low_spec_gap}
    \inf_{f\in \ELL_0(\mu),\, \Vert f \Vert_\mu=1} 
    \mathcal{E}_\mu(P,f) \geq \gamma.
 \end{equation}
Moreover, we have
\[
  1-\sup_{f\in \ELL_0(\mu),\, \Vert f \Vert_\mu=1} \langle \mathrm{P}f,f\rangle_{\mu}= \inf_{f\in \ELL_0(\mu),\, \Vert f \Vert_\mu=1} 
    \mathcal{E}_\mu(P,f),
\]
where
by the
$\mu$-reversibility (that is, self-adjointness of ${P}$) and positivity, 
the left-hand side coincides with the absolute $L^2(\mu)$-spectral gap of $P$. Consequently, \eqref{eq: low_spec_gap} implies that $\gamma$ is a lower bound of the latter.
\end{remark}

\section{Hybrid and Ideal Slice Sampling}\label{subsec:comp_SS}

For given $\sigma$-finite measure space $(\mathsf{G}, \mathcal{G}, \nu)$, assume that the distribution $\pi$ takes the form of \eqref{eq: repres_pi}, i.e. there is a not-necessarily-normalized density function $\varpi \colon \sfG \to (0,\infty)$ such that $\varpi \propto \frac{{\rm d}\pi}{{\rm d}\nu}$. 
For well-definedness purposes, we assume that the supremum norm $\Vert \varpi \Vert_\infty$ coincides with the essential supremum norm 
$\mathrm{ess}_\nu \sup \varpi$.
For $\sfT := \left (0,  \Vert \varpi \Vert_\infty \right)$ define the 
superlevel set (or slice) of $\varpi$ at $t\in\sfT$, as
\begin{align*}
    \sfG \left( t \right) := \left\{ x \in \mathsf{G}: \varpi \left( x \right) > t \right\}.
\end{align*}
This induces a function $m\colon \sfT \to (0,\infty)$ by $ m(t) := \nu \left( \sfG \left( t \right) \right)$ and a probability measure $\nu_t$ on $(\mathsf{G},\mathcal{G})$ by
\begin{equation}\label{eq:def_pi_t}
    \nu_{t} \left(A\right) := \frac{\nu\left(A\cap \sfG\left(t\right)\right)}{m \left( t \right)}, \quad A\in \mathcal{G},\; t\in \sfT.
\end{equation}
Within this framework and with the provided objects, we introduce the ideal and hybrid setting.

\parskip=1ex
\noindent
\textbf{Setting $\mathsf{Slice} \left( \pi, \nu \right)$:}
Ideal Slice Sampling in this context corresponds to 
the Markov kernel specified (by $\varpi$ and $\nu$) as
\begin{align}
    U \left( x, A \right) & =  \frac{1}{\varpi \left( x \right)} \int_{0}^{\varpi\left( x \right)} \nu_{t} \left(A\right) \dif  t, \quad x \in \mathsf{G}, A \in \mathcal{G}.
    \label{eq:U_def}
\end{align}
\begin{remark}
For the moment, let $\sfG = \R^d$ equipped with the Borel $\sigma$-algebra. Then, for $\nu$ being the Lebesgue measure, we recover `uniform' slice sampling which we considered in  Algorithm~\ref{alg:SSS}. For $\nu (\dif x) = \vert x \vert^{1-d} \, \dif x$, where $\vert \cdot \vert$ denotes the Euclidean norm in $\R^d$, we obtain polar- and more generally for $\nu (\dif x) = \vert x \vert^{k-d} \, \dif x$ with $k \in (0, \infty)$, we obtain $k$-polar slice sampling, cf. \citep{Roberts2002polar, RuSch24, Schar2023revisited}. Moreover, for $\nu$ being some Gaussian measure we have `Gaussian' slice sampling, which is closely related to elliptical slice sampling \citep{Murray2010elliptical}.
\end{remark}

\parskip=1ex
\noindent
\textbf{Setting $\mathsf{HybridSlice} \left( \pi, \nu, (H_t) \right)$:} For Hybrid Slice Sampling assume that for any $t \in \sfT$, we have a Markov kernel $H_t$ on $\sfG(t)$ which is reversible with respect to $\nu_t$.
Then, for given $\nu$, $\varpi$ and $(H_t)_{t \in \sfT}$ the corresponding Markov kernel in this setting is
\begin{equation}
    H \left( x, A \right) = \frac{1}{\varpi \left( x \right)}  \int_{0}^{\varpi \left( x \right)} H_{t} \left( x, A\cap \sfG(t)\right) \, \dif t, \quad x \in \sfG, A\in \mathcal G.
    \label{eq:H_def}
\end{equation}
Note that the setting $\mathsf{HybridSlice} \left( \pi, \nu, (\nu_t) \right)$ coincides with $\mathsf{Slice} \left( \pi, \nu \right)$, i.e. $U=H$. In this sense, we recover Ideal Slice Sampling as a special case of Hybrid Slice Sampling.

\subsection{Comparison results for slice sampling}

We start with the main regularity assumption for our comparison arguments.
\begin{assumption}\label{assms}
For a measurable function $\widetilde{\beta}: (0,\infty) \times \sfT \to [0, 1/4]$ within the setting $\mathsf{HybridSlice} \left( \pi, \nu, (H_t) \right)$ assume that for Lebes\-gue-almost every $t \in \sfT$ the Markov kernel $H_t$ satisfies
a WPI w.r.t. $\widetilde{\beta}(\cdot,t)$, i.e. 
for each $s > 0$ and $f \in \ELL_0 (\nu_t)$ we have  by Definition~\ref{def:WPI} that
    \begin{equation}
        \Vert f \Vert_{\nu_t}^2 \leq s \cdot \calE_{\nu_t} \left( H_t, f \right) + \widetilde{\beta} \left(s, t \right) \cdot \|f\|_{\osc,\nu_t}^2.
        \label{eq:Ht_WPI}
    \end{equation}
\end{assumption}
\begin{remark}
    Recall from Remark~\ref{rmk:beta} that a sufficient condition for the existence of a function $\widetilde{\beta}(t,\cdot)$ 
    to satisfy \eqref{eq:Ht_WPI} is $\nu_t$-irreducibility of $H_t$. In turn, assuming irreducibility for the on-slice Markov kernels is a fairly weak requirement for problems in which the sets $\sfG(t)$ are well-connected. 
\end{remark}
\begin{remark}
    Note that a function $f\in \ELL_0(\nu_t)$ is only supported on $\sfG(t)$, but for $f$ defined on $\sfG \supseteq \sfG(t)$, we can always consider the restriction $f|_{\sfG(t)}$ of $f$ to $\sfG(t)$.
\end{remark}
\begin{remark}
   
 Under Assumption~\ref{assms}, by Theorem~\ref{thm:conv}, the convergence $\| \mathrm{H}_t^n f\|^2_{\nu_t} \to 0$ as $n \to \infty$ for $f \in \ELL_0(\nu_t)$ with $\Vert f \Vert_{\osc,\nu_t}<\infty$ can be quantified where it may be slower than exponential.
   We emphasize that in Assumption~\ref{assms} in the WPI requirement it is hidden that $H_t$
 must be $\nu_t$-reversible and positive. 
 However, the positivity of $H_t$ is not too restrictive. Namely, for a $\nu_t$-reversible $H_t$ 
 the lazy modification $\bar{H}_t(x,A) = (H_t(x,A)+ {\bf 1}_{A}(x))/2$ for $x\in \sfG(t)$, $A\in \mathcal{G}$ with $A\subseteq \sfG(t)$ is $\nu_t$-reversible and positive.
\end{remark}

The following lemma adds desirable regularity requirements for slice sampling. For a proof of these facts we refer to \cite[Lemma~1, Lemma~5]{Latuszynski2014convergence} and \cite[Section~5]{Mira2002}.

\begin{lemma}
Within the setting of $\mathsf{Slice} \left( \pi, \nu \right)$ 
the Markov kernel $U$ is $\pi$-reversible and positive. In the setting of $\mathsf{HybridSlice} \left( \pi, \nu, (H_t) \right)$, assume that for any $t\in\sfT$ the Markov kernel $H_t$ is positive (on $\ELL(\nu_t)$). Then, the Markov kernel $H$ is $\pi$-reversible and positive. 
\end{lemma}

The latter lemma allows to compare the Dirichlet forms regarding the dominance of Markov kernels, see below \eqref{eq:dirichlet_repn}.

\begin{theorem}
    \label{thm:SS_comp}
Within the $\mathsf{HybridSlice} \left( \pi, \nu, (H_t) \right)$-setting with Markov kernel $H$, assume that for any $t\in\sfT$ the Markov kernel $H_t$ is positive. Then, the Markov kernel $U$ from setting $\mathsf{Slice} \left( \pi, \nu \right)$ dominates $H$, i.e. for all $f \in \ELL_0(\pi)$ holds
    \begin{equation}
        \calE_\pi (H,f) \leq \calE_{\pi} (U,f).
        \label{eq:H_worseU}
    \end{equation}
    More restrictively, let Assumption~\ref{assms} be satisfied w.r.t. $\widetilde{\beta}$. Then, for all $s>0$, $f \in \ELL_0(\pi)$ we have
    \begin{equation}
        \calE_\pi (U,f) \leq s \cdot \calE_\pi (H,f) + \beta(s) \cdot \|f\|_{\osc,\pi}^2,
        \label{eq:U_WPI_H}
    \end{equation}
    where $\beta : (0,\infty)\to [0, \infty)$, given by
    $
        \beta(s) := c^{-1} \cdot \int_\sfT \widetilde{\beta}(s,t)\, m \left( t \right) \, \dif t,
    $
    is decreasing and satisfies $\lim_{s\to\infty} \beta(s)=0$. 
\end{theorem}
The theorem can be interpreted and read as follows: the Dirichlet form, as a measure of variance dissipation (where larger values correspond to faster convergence), entails with \eqref{eq:H_worseU} that a Markov chain with Markov kernel $H$ converges slower to $\pi$ than a Markov chain with Markov kernel $U$. Examining \eqref{eq:U_WPI_H}, we can then quantify the `slow down' which $H$ suffers, relative to $U$.

 We add consequences under additional assumptions that offer different perspectives around the result.

\begin{corollary}
\label{cor:Ht_gap_comp}
For a measurable function $\gamma\colon \sfT \to (0,\infty)$, let
Assumption~\ref{assms} be satisfied with
\[
\widetilde{\beta}(s,t) = 1/4\cdot \mathbf{1}_{(0,\gamma(t)^{-1})}(s),
\]
i.e. by Remark~\ref{rem: SPI} for any $t\in\sfT$ the Markov kernel $H_t$ satisfies a Poincar\'e inequality and the number $\gamma(t)$ is a lower bound of the corresponding absolute $\ELL(\nu_t)$-spectral gap.
Then, \eqref{eq:U_WPI_H} holds with
    \begin{equation*}
        \beta(s) = \left(4 \cdot c \right)^{-1} \cdot \int_\sfT \mathbf{1}_{(0,\gamma(t)^{-1})}(s) \cdot m \left( t \right) \dif t.
    \end{equation*}
    Moreover, if $\gamma_\star := \inf_{t\in\sfT} \gamma(t)>0$, then for all $f\in\ELL_0(\pi)$ we have
    \[
               \gamma_\star\cdot \calE_{\pi} (U,f) \leq \calE_\pi (H,f) \leq \calE_{\pi} (U,f).
    \]
\end{corollary}
\begin{proof}
    As an immediate consequence of Theorem~\ref{thm:conv} it follows that $\beta$ satisfies \eqref{eq:U_WPI_H}. In the case of $\gamma(t)\geq \gamma_\star>0$ for all $t\in\sfT$ we have
    \[
    \beta(s)\leq 1/4\cdot \mathbf{1}_{(0,\gamma_\star^{-1})}(s),
    \]
    where we used $c=\int_{\sfT} m(t) \dif t$, cf. the initial arguments in \cite[proof of Lemma~1]{Latuszynski2014convergence}. Now, for $s=\gamma_\star^{-1}$ in \eqref{eq:U_WPI_H} the lower and from \eqref{eq:H_worseU} the upper bound w.r.t.  $\calE_\pi (H,f)$ follows.
\end{proof}

Corollary~\ref{cor:Ht_gap_comp} can be interpreted as follows.
On the one hand, if the kernels $H_t$ admit a uniform lower bound of the $\ELL(\nu_t)$-spectral gaps, i.e. $\inf_{t \in \sfT} \gamma(t) = \gamma_\star>0$, we conclude by Remark~\ref{rem: SPI} that $H$ possesses an $\ELL(\pi)$-spectral gap `relative to' $U$. That is, one can genuinely view the hybrid sampler as being at least as efficient as the comparable ideal one, slowed down by a constant factor. On the other hand, if `sufficiently many' on-slice kernels $H_t$ have only slower than exponential convergence to their invariant measure -- in the sense that $\widetilde{\beta} \left( s, t \right) > 0$ for all $s>0$ to the extent that the function $\beta(s)$ is strictly positive for all $s>0$ -- then we can only deduce a slower-than-exponential rate of convergence for $H$ (relative to $U$). In fact, even if all on-slice kernels $H_t$ possess a spectral gap, but these gaps tend to $0$ in some limit, then our results only allow us to deduce that $H$ has slower than exponential convergence behaviour relative to $U$. 

\begin{remark}
    Theorem~\ref{thm:SS_comp} simplifies and extends \cite[Theorem~8]{Latuszynski2014convergence} and the recent improvement \cite[Proposition~28]{Qin2023}. In particular, i) the somewhat opaque $\beta_k$ condition of \cite[Theorem~8]{Latuszynski2014convergence} is removed, and ii) we are able to make nontrivial deductions about the subgeometric case, accommodating the scenario in which either the Ideal Slice Sampling kernel or the on-slice kernels $(H_t)_{t \in \sfT}$, or neither possess an absolute spectral gap.
\end{remark}

\begin{remark}
    A similar result was proven in \cite[Theorem 25]{Andrieu2018uniform} for two-component Gibbs samplers. Indeed, Ideal Slice Sampling admits an interpretation in terms of a two-component Gibbs sampler over the joint state $(x, t) \in \sfG\times\sfT$. This comparison result above can then also be seen as an extension of \cite[Theorem 25]{Andrieu2018uniform} to the subgeometric scenario.
\end{remark}

Note that \cite{Natarovskii2021quantitative,RuSch24} provide sufficient conditions in the $\mathsf{Slice}(\pi,\nu)$-setting which imply that $U$ possesses a strictly positive $\ELL(\pi)$-spectral gap, cf. Remark~\ref{rem: SPI}. We add a consequence of Theorem~\ref{thm:conv} when this is indeed the case.

\begin{corollary} \label{cor: spec_gap_U_implies_WPI_H}
    Under Assumption~\ref{assms} w.r.t. $\widetilde{\beta}$, and provided that $U$ satisfies for some constant $\gamma > 0$, for all $f\in \ELL_0(\pi)$ that
\begin{equation}
     \gamma \cdot \|f\|^2_\pi \leq \calE_\pi (U,f),
     \label{eq:SPI_SSS}
\end{equation}
     i.e $U$ admits a (strong) Poincar\'e inequality, a.k.a. possesses $\gamma$ as a lower bound of the $\ELL(\pi)$-spectral gap, cf. Remark~\ref{rem: SPI}, 
     we have 
     for all $s>0$ and $f\in \ELL_0(\pi)$, that
    \begin{equation}
    \label{eq:WPI_H}
        \|f\|^2_\pi \leq s\cdot \calE_\pi (H,f) + \bar \beta(s) \cdot \|f\|_{\osc,\pi}^2,
    \end{equation}
    with
    \begin{equation*}
        \bar{\beta}(s):=c^{-1} \gamma^{-1} \int_{\sfT} \widetilde{\beta}(\gamma s,t)\, m(t)\dif t
        = 
        \gamma^{-1} \, \beta(\gamma \cdot s),
    \end{equation*}
    and $\beta$ from  Theorem~\ref{thm:SS_comp}.
   In particular, $H$ satisfies a WPI w.r.t. $\bar \beta$.
\end{corollary}
\begin{proof}
Inequality \eqref{eq:WPI_H} follows by Theorem~\ref{thm:SS_comp}, \eqref{eq:SPI_SSS} and a rescaling argument. Then, by the inheritance of properties of $\beta$ to $\bar\beta$, it is implied that $H$ fulfills the WPI requirements.
\end{proof}
An immediate implication of the former corollary, by applying Theorem~\ref{thm:conv}, is the deduction of an overall (possibly slower-than-exponential) convergence bound for $\|\mathrm{H}^n f\|_2^2$ for any $f\in \ELL_0(\pi)$.
\subsection{Proof of Theorem~\ref{thm:SS_comp}}
We begin with two lemmas that provide useful representations of the Dirichlet forms of $H$ and $U$.

\begin{lemma}\label{lemma:EH_repn}
   Within the $\mathsf{HybridSlice} \left( \pi, \nu, (H_t) \right)$-setting the Dirichlet form of $H$ satisfies for any $f\in\ELL(\pi)$ that
   \begin{equation*}
        \calE_\pi (H,f) = c^{-1} \int_\sfT \calE_{\nu_t} (H_t, f)  m(t) \, \dif t,
    \end{equation*}
    where for each $t \in \sfT$,
    \begin{equation*}
        \calE_{\nu_t} (H_t, f) = \frac{1}{2} \int_{\mathsf{G}(t)} \int_{\mathsf{G}(t)} \left[ f(x)-f(y) \right]^2 H_t( x, \dif y)  \nu_t(\dif x)<\infty,
    \end{equation*}
    is the Dirichlet form of $H_t$ (which coincides with $\calE_{\nu_t} (H_t, f|_{\mathsf{G}(t)})$).
\end{lemma}

\begin{proof}
     Since $\int_{\mathsf{G}(t)} f^2(x) \, \nu_t (\dif x)  \leq \frac{c}{t \,m(t)}  \int_{\mathsf{G}(t)} f^2(x) \,\pi (\dif x) \leq \frac{c}{t\, m(t)} \|f\|^2_\pi < \infty$, observe that $\mathcal{E}_{\nu_t} (H_t, f) < \infty$ for every $t\in \sfT$. 
    For convenience, we write $g(x,y) := [f(x)-f(y)]^2$ and have by \eqref{eq:dirichlet_repn} that
    \begin{equation*}
        \begin{split}
            2 \, \calE_\pi (H,f)&=  \int_{\mathsf{G}} \int_{\mathsf{G}} g(x,y)  H \left( x, \dif y \right) \, \pi(\dif x)\\ 
            &= c^{-1}  \int_{\mathsf{G}}  \int_0^{\varpi(x)} \int_{\mathsf{G}(t)} \,g(x,y)  H_t \left( x, \dif y \right) \, \dif t \, \nu (\dif x) .
        \end{split}
    \end{equation*}
    Changing the order of integration yields
    \begin{equation*}
        \begin{split}
            & c^{-1} \int_{\mathsf{G}}  \int_0^{\varpi(x)}\int_{\mathsf{G}(t)} \,g(x,y)  H_t \left( x, \dif y \right) \dif t \, \nu(\dif x) 
            \\&= c^{-1}  \int_0^{\|\varpi\|_\infty} \int_{\mathsf{G}(t)} \int_{\mathsf{G}(t)} g(x,y) \, H_t \left( x, \dif y \right)  \nu(\dif x)\dif t\\
            &= c^{-1} \int_\sfT \int_{\mathsf{G}(t)} \int_{\mathsf{G}(t)} g(x,y) H_t \left( x, \dif y \right)  \nu_t(\dif x)  \,m(t)\dif t \\
            &= 2 c^{-1} \int_\sfT \calE_{\nu_t} (H_t, f) m(t)\dif t,
        \end{split}
    \end{equation*}
    where we used the definition of $\nu_t$ from \eqref{eq:def_pi_t}.
\end{proof}

As remarked earlier, Ideal Slice Sampling with transition kernel $U$ coincides with Hybrid Slice Sampling whenever $H_t  = \nu_t $ for all $t\in \sfT$. Therefore, by the previous lemma we have the following representation of $\calE_\pi (U,f)$.
\begin{lemma}\label{cor:SS_rep}
Within the $\mathsf{Slice} \left( \pi, \nu \right)$-setting with Markov kernel $U$, we have for any $f\in \ELL(\pi)$
that
    \begin{equation*}
        \calE_\pi (U,f)= c^{-1}  \int_\sfT \Var_{\nu_t} (f) \, m(t)\,\dif t.
    \end{equation*}
\end{lemma}

\begin{proof}
    Lemma~\ref{lemma:EH_repn} implies
    \begin{equation*}
        \calE_\pi (U,f) = \frac{1}{2 c} \int_\sfT  \int_{\mathsf{G}(t)}  \int_{\mathsf{G}(t)} \left[ f(x)-f(y) \right]^2  \nu_t(\dif y)\,  \nu_t(\dif x)\,  m(t)\dif t.
    \end{equation*}
    Moreover
    \begin{equation*}
        \frac{1}{2} \int_{\mathsf{G}(t)}  \int_{\mathsf{G}(t)} \left[ f(x)-f(y) \right]^2  \nu_t(\dif y)  \nu_t(\dif x) = \nu_t(f^2) - \nu_t(f)^2 =\Var_{\nu_t}(f),
    \end{equation*}
    which finishes the proof.
\end{proof}

We now establish the two comparison inequalities which comprise Theorem~\ref{thm:SS_comp}. We start with the proof of \eqref{eq:H_worseU}. For the convenience of the reader we repeat the statement and the requirements in the following proposition.

\begin{proposition} \label{prop:HleU}
Within the $\mathsf{HybridSlice} \left( \pi, \nu, (H_t) \right)$-setting with Markov kernel $H$, assume that for any $t\in\sfT$ the Markov kernel $H_t$ is positive. 
Then, for the Markov kernel $U$ which is obtained by setting $\mathsf{Slice} \left( \pi, \nu \right)$, for all $f\in \ELL_0(\pi)$ we have 
    \begin{equation*}
        \calE_\pi (H,f) \leq \calE_\pi (U,f).
    \end{equation*}
\end{proposition}
\begin{proof}
   For any $t\in\sfT$  and any $f\in \ELL_0(\nu_t)$ it follows
    by the positivity of $H_t$ that
    \begin{equation*}
        \calE_{\nu_t} (H_t, f) = \langle f, (\Id - H_t) f \rangle_{\nu_t} = \|f\|^2_{\nu_t} - \langle f, H_t f\rangle_{\nu_t} \leq \|f\|^2_{\nu_t}.
    \end{equation*}
     Therefore, by applying Lemma~\ref{lemma:EH_repn}, using $\mathcal{E}_{\nu_t}(H_t,f)= \mathcal{E}_{\nu_t}(H_t,f-\nu_t(f))$ as well as 
    $\Var_{\nu_t}(f) = \Vert f - \nu_t(f) \Vert_{\nu_t}^2$
    and
    Corollary~\ref{cor:SS_rep}
    the assertion is proven.
\end{proof}
We now turn to the proof of \eqref{eq:U_WPI_H}. Again, we repeat the formulation of the statement in the next proposition.
\begin{proposition}     \label{prop:UleH}
Let Assumption~\ref{assms} be satisfied w.r.t. the mapping $\widetilde{\beta}$. Then, with $U$ and $H$ specified in the settings $\mathsf{Slice} \left( \pi, \nu \right)$ and $\mathsf{HybridSlice} \left( \pi, \nu, (H_t) \right)$  respectively, for all $s>0$ and $f\in \ELL(\pi)$ we have
    \begin{equation*}
        \calE _\pi (U,f) \leq s\cdot \calE_\pi (H,f) + \beta(s) \|f\|_{\osc,\pi}^2,
    \end{equation*}
    where $\beta : (0, \infty) \to [0, \infty)$ is given by the map $ s\mapsto \beta(s) :=  c^{-1}  \int_\sfT \widetilde{\beta}(s,t)  m(t) \, \dif t$.
\end{proposition}
\begin{proof}
    For any $t\in\sfT$ and any $f\in \ELL(\pi)$ observe that by the absolute continuity of $\nu_t$ w.r.t. $\pi$ we have $\Vert f \Vert_{\osc,\nu_t} \leq \Vert f \Vert_{\osc,\pi}$.
    Then, by
    Corollary~\ref{cor:SS_rep}, we obtain
    \begin{equation*}
        \begin{split}
            \calE_\pi (U,f) &= c^{-1}  \int_\sfT  \Var_{\nu_t} (f) \, m(t) \, \dif t\\
            &\leq c^{-1}  \int_\sfT  \left [ s \,\calE_{\nu_t} (H_t, f) + \widetilde{\beta}(s, t) \, \|f\|_{\osc,\nu_t}^2\right ] m(t) \dif t\\
            &= s \cdot \calE_\pi (H,f)+ \|f\|_{\osc,\pi}^2 \cdot c^{-1} \int_\sfT 
            \widetilde{\beta}(s,t) \,m(t) \dif t,
        \end{split}
    \end{equation*}
    for any $s>0$, where we have used \eqref{eq:Ht_WPI} and Lemma~\ref{lemma:EH_repn}.
\end{proof}

Finally, regarding Theorem~\ref{thm:SS_comp} it remains to verify the claimed properties of the function $\beta$.
\begin{proposition}
        Let Assumption~\ref{assms} be satisfied w.r.t. the mapping $\widetilde{\beta}$.
        Then, the function $\beta\colon (0,\infty) \to [0,\infty)$ given by $s \mapsto \beta(s) :=  c^{-1} \cdot \int_\sfT \widetilde{\beta}(s,t) \, m(t) \dif t$ is decreasing, and satisfies $\lim_{s\to \infty} \beta(s) = 0$.
    \label{lem:beta_conds}
\end{proposition}

\begin{proof}
    Since $s \mapsto \widetilde{\beta}(s,t)$ is monotone decreasing for each $t \in \sfT$, it follows that $\beta$ is also monotone decreasing. Furthermore, since $\widetilde{\beta} \left( \cdot, t \right) \leq 1/4$ pointwise for all $t \in \sfT$ as emphasized in Remark~\ref{rmk:beta}, and $c^{-1} \cdot \int_\sfT m(t) \dif t =\int_{\mathsf{G}} \pi (\dif x) = 1$, and each $\widetilde{\beta}(s,t)\to 0$ as $s\to \infty$, the dominated convergence theorem yields that
    \begin{equation*}
        \lim_{s\to \infty} \beta(s) = c^{-1} \int_\sfT \lim_{s\to \infty} \widetilde{\beta} \left( s, t \right) \, m(t) \dif t=0.
        \qedhere
    \end{equation*}
\end{proof}
For completeness, we emphasize that combining all three previous propositions proves Theorem~\ref{thm:SS_comp}.

\section{Examples}\label{sec:examples}
Our abstract comparison result of Theorem~\ref{thm:SS_comp} can be further specialised to provide novel analyses of particular algorithms. In this section, we detail various examples to this effect.

\subsection{Metropolis as Hybrid Slice Sampling}

We retain the general regularity requirements of Section~\ref{subsec:comp_SS}, that is, $(\sfG,\mathcal{G},\nu)$ is a $\sigma$-finite measure space, $\varpi \colon \sfG \to (0,\infty)$ with $\varpi \propto \frac{{\rm d}\pi}{{\rm d}\nu}$ for target measure $\pi$ and $\Vert \varpi \Vert_\infty = \mathrm{ess}_\nu \sup \varpi$ by assumption. We start with specifying a general Metropolis setting:

\parskip=1ex
\noindent
\textbf{Setting $\mathsf{Metropolis} \left( \pi, \nu, Q\right)$:}
For $Q\colon \sfG \times \mathcal{G} \to [0,1]$ being a $\nu$-reversible
proposal Markov kernel, define the Metropolis Markov kernel $M\colon \sfG \times \mathcal{G} \to [0,1]$ by 
\begin{align*}
    M \left(x, A\right) = \int_{A} a \left( x, y \right) Q \left(x, {\rm d}y\right) 
     + {\bf 1}_{A} \left( x \right) \cdot \bar{\alpha} \left( x \right), \quad x \in \mathsf{G}, A \in \mathcal{G},
\end{align*}
where for $x, y\in\mathsf{G}$ we set
\begin{align}
\label{al: acc_prob}
    a\left(x, y\right) := \min\left\{1 , \frac{\varpi\left(y\right)}{\varpi\left(x\right)}\right\}, \qquad \bar{\alpha} \left(x\right) := 1 - \int_{\mathsf{G}} a \left( x, y \right)\,Q \left(x, {\rm d}y\right).
\end{align}

It is well known that $M$ is $\pi$-reversible. We describe example settings, essentially derived by concrete choices of $Q$.

\begin{example}[Setting $\mathsf{RWM}(\pi,Q)$]
\label{ex: RWM}

Here $\sfG\subseteq \mathbb{R}^d$, $\mathcal{G}=\mathcal{B}(\sfG)$ is the Borel $\sigma$-algebra and $\nu=\mathrm{Leb}$ denotes the $d$-dimensional Lebesgue measure. For $q\colon \sfG\times\sfG \to [0,\infty)$ let $Q(x,\dif y) := q(x,y) \nu(\dif y)$ be a $\nu$-reversible Markov kernel. Then, the \emph{Random Walk Metropolis} setting, denoted as $\mathsf{RWM}(\pi,Q)$, is defined by the 
$\mathsf{Metropolis}\left(\pi,\mathrm{Leb},Q\right)$-setting. 
Moreover, a conventional choice for $Q$ is $Q_{\sigma} \left( x, \dif y\right) = q_\sigma(x,y) \dif y$
with $q_\sigma(x,y) \propto \exp(-\vert x-y\vert^2/(2\sigma^2))$, where $\vert \cdot \vert$ denotes the Euclidean norm and $\sigma>0$ a step-size parameter.
Let us emphasize that  $Q_\sigma$ considered on $\ELL_{\nu}$ is positive, cf. \cite{Doucet2015intractable}.

\end{example}

\begin{example}[Setting $\mathsf{IM}(\pi, \nu) $]
\label{example:IMH}
Here $(\sfG,\mathcal{G},\nu)$ is a general probability space and $Q \left( x, \cdot \right) = \nu$, i.e. $Q$ does not depend on $x\in\sfG$. Then, the \emph{Independent Metropolis} setting  with `proposal' $\nu$, denoted as $\mathsf{IM}(\pi, \nu)$, 
is defined by the
$\mathsf{Metropolis} \left(\pi, \nu, \nu \right)$-setting. By taking  \cite[]{Doucet2015intractable} (cf. also \cite[end of page~3]{Gasemyr2006spectrum}) into account it follows that $Q$ on $\ELL_{\nu}$ and the corresponding $\pi$-reversible Metropolis Markov kernel $M$ is also positive.
\end{example}

\begin{example}[Setting $\mathsf{pCN}(\pi,\mathsf C, s)$]
Here let $\sfG$ be a possibly infinite-di\-men\-sional separable Hilbert space and $\mathcal{G}=\mathcal{B}(\sfG)$ be the corresponding Borel $\sigma$-algebra. Assume that $\nu=\mathcal{N}(0,\mathsf C)$ is a Gaussian measure on $(\sfG,\mathcal{G})$ with $\mathsf C\colon \sfG \to \sfG$ being a nonsingular covariance operator. For step-size parameter $s\in[0,1]$, set $Q_{s,\mathsf{C}}(x,\cdot):= \mathcal{N}(\sqrt{1-s^2}x,s^2 \mathsf{C})$. Then, the \emph{Preconditioned Crank-Nicolson Metropolis} setting, denoted as $\mathsf{pCN}(\pi,\mathsf C,s)$, is defined by the $\mathsf{Metropolis} \left( \pi, \mathcal{N}(0,\mathsf{C}), Q_{s,\mathsf{C}} \right)$-setting, see e.g. \citep{cotter2013mcmc}.
Moreover, it is known -- see \cite[Lemma~9, Theorem~11 with $\Gamma=0$]{rudolf2018generalization} and also \cite{Baxendale2005geometric} -- that $Q_{s,\mathsf{C}}$ on $\ELL_{\nu}$ as well as the corresponding $\pi$-reversible Metropolis Markov kernel $M$ is positive. 
\end{example}

As observed in \cite[Section~2.1]{higdon1998auxiliary}, in a finite state space setting, Metropolis Markov kernels can be seen as instances of Hybrid Slice Samplers. We extend this statement to general state spaces and add a comparison result.

\begin{proposition}
\label{lem:metropolis-are-hss}
Within the $\mathsf{Metropolis} \left(\pi, \nu, Q \right)$-setting, with resulting Metropolis Markov kernel $M$,
let the $\mathsf{HybridSlice} \left(\pi, \nu, (H_t) \right)$-setting, with resulting Mar\-kov kernel $H$, be specified by  
\begin{align*}
    H_{t} \left(x, A\right) = Q \left(x, A\cap \sfG \left(t \right) \right) + \mathbf{1}_{A} \left(x\right) (1-Q(x,\sfG(t))), \quad t\in\sfT, x\in\sfG(t), A\in\mathcal{G}.
\end{align*}
Then $M=H$, that is, the settings coincide. Moreover, regarding the setting $\mathsf{Slice} \left( \pi, \nu \right)$, with resulting Markov kernel $U$, we have for positive $Q$ (on $\ELL_{\nu}$) that $U$ dominates $M$, i.e. $\mathcal{E}_\pi(M,f) \leq \mathcal{E}_\pi(U,f)$ for all $f\in \ELL_0(\pi)$. 
\end{proposition}

\begin{proof}
For $x\in \sfG$ and $A\in\mathcal{G}$ we have
\begin{align*}
   & H\left(x,A\right)  =
   \frac{1}{\varpi \left( x \right)}
   \int_{0}^{\varpi \left( x \right)} H_{t} \left( x, A \right) \, \dif t\\
        =& \frac{1}{\varpi \left( x \right)} \int_{0}^{\varpi \left( x \right)} Q \left( x, A \cap G \left( t \right) \right) \, \dif t+  \frac{\mathbf{1}_{A} \left( x \right)}{\varpi \left( x \right)} \int_{0}^{\varpi \left( x \right)} Q ( x, 
        \sfG \left(t\right)^{\mathsf{c}}) \, \dif t \\
        =& \frac{1}{\varpi \left( x \right)}  \int_{0}^{\varpi \left( x \right)} \int_{A} \mathbf{1}_{\sfG \left( t \right)} \left( y \right)
        Q \left( x, \dif y \right)  \, \dif t
        + \frac{\mathbf{1}_{A} \left( x \right)}{\varpi \left( x \right)} \int_{0}^{\varpi \left( x \right)} \int_{\mathsf{G}}  \mathbf{1}_{\sfG \left( t \right)^{\mathsf{c}}} \left( y \right) Q \left( x, \dif y \right) \dif t\\
        =& \frac{1}{\varpi \left( x \right)} \int_{0}^{\varpi \left( x \right)} \int_{A} \mathbf{1}_{(0,\varpi(y))}(t) \, Q \left( x, \dif y \right) \, \dif t\\
        &\qquad \qquad + \frac{\mathbf{1}_{A} \left( x \right)}{\varpi \left( x \right)} \int_{0}^{\varpi \left( x \right)} \int_{\mathsf{G}}  \mathbf{1}_{[\varpi(y),\infty)} (t) \, Q \left( x, \dif y \right)  \dif t\\
        =& \frac{1}{\varpi \left( x \right)} \int_{\sfT} \int_{A} \mathbf{1}_{(0,\min\{\varpi(x),\varpi(y)\})} (t) \,Q \left( x, \dif y \right)  \dif t\\
        &\qquad \qquad + \frac{\mathbf{1}_{A} \left( x \right)}{\varpi \left( x \right)} \int_{\sfT} \int_{\mathsf{G}}  \mathbf{1}_{[\varpi(y),\varpi(x))}(t)  Q \left( x, \dif y \right)\dif t\\
        =& \frac{1}{\varpi \left( x \right)}\int_{A}   \min\{ \varpi \left( x \right), \varpi \left( y \right)\}\,  Q \left( x, \dif y \right)\\
        &\qquad \qquad + \frac{\mathbf{1}_{A} \left( x \right)}{\varpi \left( x \right)}  \int_{\mathsf{G}} ( \varpi \left( x \right) - \min \{ \varpi \left( x \right), \varpi \left( y \right) \})Q \left( x, \dif y \right)\\
        =& \int_{A} a \left( x, y \right) Q \left( x, \dif y \right) + \mathbf{1}_{A} (x ) \cdot \bar{\alpha} \left( x \right) = M(x,A).
\end{align*}
By the fact that $Q$ is $\nu$-reversible and positive, it follows that $H_t$ is $\mu_t$-reversible and positive on $\ELL(\nu_t)$ for any $t\in\sfT$. Finally, Theorem~\ref{thm:SS_comp} implies the dominancy statement.
\end{proof}

Considering the scenarios of the examples of this section allows us to conclude the following comparison statements.

\begin{corollary}
    \begin{enumerate}
        \item 
        Metropolis Markov kernel $M$ within the setting 
        $\mathsf{RWM} \left( \pi, Q\right)$, with $Q=Q_\sigma$ from Example~\ref{ex: RWM} for $\sigma>0$, is dominated by Markov kernel $U$ from the corresponding setting 
        $\mathsf{Slice}  \left(\pi,\mathrm{Leb}\right)$.
        \item 
        Metropolis Markov kernel $M$ within the setting $\mathsf{IM} \left(\pi, \nu \right)$  is dominated by Markov kernel $U$ from the corresponding setting $\mathsf{Slice} \left(\pi, \nu \right)$.
        \item Metropolis Markov kernel $M$ within the setting $\mathsf{pCN}(\pi,\mathsf C, s)$ is dominated by Markov kernel $U$ from the corresponding setting $\mathsf{Slice} \left( \pi, \mathcal{N}(0,\mathsf{C})\right)$.
    \end{enumerate}
\end{corollary}

Elaborating on the application of Theorem~\ref{thm:SS_comp} enables us to partially reverse these comparisons, demonstrating that in certain cases, Hybrid Slice Samplers may be almost as efficient as their ideal analogs.

\subsubsection{Independent Metropolis -- Example~\ref{example:IMH} reconsidered}

We illustrate the application of \eqref{eq:U_WPI_H} of Theorem~\ref{thm:SS_comp} and discuss its possible limitations. 
Recall that in the $\mathsf{IM} \left( \pi, \nu \right)$-setting of Example~\ref{example:IMH} a probability space $(\sfG,\mathcal{G},\nu)$ is present. 
The corresponding Metropolis Markov kernel is given as
\begin{align*}
    M \left( x, A \right) &= \int_A a(x,y) \nu \left( \dif y \right)  + \mathbf{1}_A \left( x \right) \cdot \left( 1-\int_{\sfG} a(x,y) \,\nu(\dif y) \right), \quad x \in \mathsf{G}, A \in \mathcal{G},
\end{align*}
with $a(x,y)$ as defined in \eqref{al: acc_prob}.
From Proposition~\ref{lem:metropolis-are-hss} we know that $M$ coincides with a Markov kernel $H$ based on suitable Hybrid Slice Sampling. In this framework we verify Assumption~\ref{assms} for the corresponding sequence of Markov kernels $(H_t)_{t\in\sfT}$. The result can be achieved by applying \cite[Proposition~27]{Andrieu2021comparison}, where we interpret $H_t$ with $t\in\sfT$ as suitable independent Metropolis kernel. However, for the convenience of the reader we prove the result directly and add an immediate consequence.

\begin{proposition}
\label{prop: IMH_gen_comp}
    For $t\in\sfT$ let Markov kernel $H_t$ on $\sfG(t)$ be given as
    \[
        H_t(x,A) = \nu(A\cap \sfG(t)) + \mathbf{1}_A(x)(1-\nu(\sfG(t))), \quad x\in \sfG(t), A\in\mathcal{G}.
    \]
    Then, $(H_t)_{t\in\sfT}$ satisfies Assumption~\ref{assms} with $\widetilde{\beta}(s,t) = 1/4\cdot \mathbf{1}_{(0,m(t)^{-1})}(s) $. Consequently, for $U$ from the $\mathsf{Slice}  \left(\pi,\nu\right)$-setting we have for all $s>0$ and $f\in \ELL_0(\pi)$ that
 \begin{align}
 \label{al: comp_IM_SS}
    \calE_\pi (M,f) \le \calE_\pi (U,f)\leq s\, \calE_\pi (M,f) + \frac{\|f\|^2_{\osc,\pi}}{4c}   \int_{\sfT} \mathbf{1}_{(0,m(t)^{-1})}(s) \, m \left( t \right) \dif t.
\end{align}
\end{proposition}
\begin{proof}
    By Proposition~\ref{lem:metropolis-are-hss}, $H$ from setting $\mathsf{HybridSlice} \left(\pi, \nu, (H_t) \right)$ coincides with $M$. Moreover, for any $t\in\sfT$ Markov kernel $H_t$ is positive and reversible, such that the first inequality of \eqref{al: comp_IM_SS} follows. For verifying \eqref{eq:Ht_WPI} we distinguish two cases.
    Firstly, let $m(t)^{-1}\leq s$. Then, for all $g\in \ELL_{0}(\nu_t)$, we have 
    \begin{align*}
        \Vert g \Vert_{\nu_t}^2 
        = \frac{1}{2} \int_{\sfG(t)}\int_{\sfG(t)}
        (g(x)-g(y))^2 \nu_t(\dif y)\nu_t(\dif x)
        \leq s \cdot \mathcal{E}_{\nu_t}(H_t,g).
    \end{align*}
    Secondly, let $m(t)^{-1}> s$. Then, for all $g\in \ELL_{0}(\nu_t)$, by  Popoviciu's inequality we have $\Vert g \Vert_{\nu_t}^2 \leq 1/4 \cdot \Vert g \Vert^2_{\osc, \nu_t}$. 
    Taking both cases together yields \eqref{eq:Ht_WPI} with $\widetilde{\beta}(s,t) = 1/4\cdot \mathbf{1}_{(0,m(t)^{-1})}(s) $. Finally, applying Corollary~\ref{cor:Ht_gap_comp} proves the second inequality of \eqref{al: comp_IM_SS} and the statement is shown.
\end{proof}

\begin{example}
    \label{ex:IMH_exp}
   We consider $\sfG=[0,\infty)$, equipped with the corresponding Borel $\sigma$-algebra $\mathcal{G}$ and exponential distribution $\nu$ with parameter $\lambda>0$, i.e. $\nu(\dif x)=\lambda e^{-\lambda x} \, \dif x$. Moreover, let $\pi$ also be an Exponential distribution with parameter $1$ on $(\sfG,\mathcal{G})$. Then, $\pi(\dif x) =  e^{-x}\, \dif x = \lambda^{-1} \,e^{-(1-\lambda)x} \nu(\dif x)$, such that $\varpi(x)=e^{-(1-\lambda)x}$. We distinguish two cases: 
   
Firstly, let $\lambda\in(0,1)$. Then, from Proposition~\ref{prop: Exp_USS_spec_gap} of Appendix~\ref{app:exp_IMH} we know for  all $f\in \ELL_0(\pi)$ that
    $
    \mathcal{E}_\pi(U,f) \geq \Vert f \Vert_\pi^2
        (1+\lambda)/2.
    $
    Therefore, by Proposition~\ref{prop: IMH_gen_comp} we have
    \[ \Vert f \Vert_\pi^2
\leq \frac{2s}{(1+\lambda)}\, \calE_\pi (M,f) + \frac{\|f\|^2_{\osc,\pi}}{2(1+\lambda)c}   \int_0^{1} \mathbf{1}_{(0,m(t)^{-1})}(s) \, m \left( t \right) \dif t.\]
By the fact that $m(t)=1-t^{\lambda/(1-\lambda)}$ we consider $b\colon(0,\infty) \to (0,\lambda]$, given by
 \begin{align*}
  b(s) & :=  \int_0^1  \mathbf{1}_{(0,m(t)^{-1})}(s) \, m \left( t \right) \dif t
     = \int_{(\max\{0,1-s^{-1}\})^{(1-\lambda)/\lambda}}^{1} (1-t^{\lambda/(1-\lambda)})\dif t\\
    & = \Big[t-(1-\lambda)t^{(1-\lambda)^{-1}}\Big]_{(\max\{0,1-s^{-1}\})^{(1-\lambda)/\lambda}}^{1}\\
    & = \lambda - \max\{0,1-s^{-1}\})^{(1-\lambda)/\lambda} + (1-\lambda) \max\{0,1-s^{-1}\})^{1/\lambda}\\
    & = \begin{cases}
          \lambda - (1-s^{-1})^{(1-\lambda)/\lambda}
          + (1-\lambda) (1-s^{-1})^{1/\lambda},
          & s>1,\\
          \lambda,                  & s\in (0,1],
        \end{cases}
 \end{align*}
 and have with $c=\lambda$ as well as a rescaling argument that $M$ from the $\mathsf{IM} \left( \pi, \nu \right)$-setting satisfies a WPI w.r.t. 
 \[
    \beta(s) := (2(1+\lambda)\lambda)^{-1} b((1+\lambda)s/2).
 \]
 
 Secondly, assume that $\lambda\in [1,\infty)$. Then, again from Proposition~\ref{prop: Exp_USS_spec_gap} of Appendix~\ref{app:exp_IMH} we know for  all $f\in \ELL_0(\pi)$ that
    $
    \mathcal{E}_\pi(U,f) \geq \Vert f \Vert_\pi^2
        \big(2\lambda-1\big)^2.
    $
    Therefore, by Proposition~\ref{prop: IMH_gen_comp} we have
    \[ \Vert f \Vert_\pi^2
\leq \frac{s}{(2\lambda-1)^2}\, \calE_\pi (M,f) + \frac{\|f\|^2_{\osc,\pi}}{4(2\lambda-1)c}   \int_0^{\infty} \mathbf{1}_{(0,m(t)^{-1})}(s) \, m \left( t \right) \dif t.\]
By the fact that $m(t)=(\max\{1,t\})^{\lambda/(1-\lambda)}$ we consider 
 \begin{align*}
   \int_0^\infty  \mathbf{1}_{(0,m(t)^{-1})}(s) \, m \left( t \right) \dif t
  &   = \mathbf{1}_{(0,1)}(s) +\int_{s^{-(1-\lambda)/\lambda}}^{\infty} t^{\lambda/(1-\lambda)}\dif t\\
  &   =  \mathbf{1}_{(0,1)}( s)
    + (\lambda-1)s^{-1/\lambda}
 \end{align*}
 and have with $c=\lambda$ as well as a rescaling argument that $M$ from the $\mathsf{IM} \left( \pi, \nu \right)$-setting satisfies a WPI w.r.t. 
 \begin{equation}
 \label{eq: beta_IMH}
       \beta(s) := \frac{1}{4(2\lambda-1)\lambda}\left(
    \mathbf{1}_{(0,(2\lambda-1)^{-2})}(s)+\frac{(\lambda-1)s^{-1/\lambda}}{(2\lambda-1)^{2/\lambda}} 
    \right).
 \end{equation}
\end{example}

\begin{remark} In the previous example, in the case $\lambda\in[1,\infty)$, the function $\beta$ behaves for $s>1$ as $s\mapsto s^{-1/\lambda}(\lambda-1)/(4\lambda)$, which for $s\to\infty$ decays to zero at a slower rate than $\bar{\beta}(s) = s^{-1/(\lambda-1)}$. From \cite[Proposition~27, Lemma~28]{Andrieu2021comparison}, we know that within the $\mathsf{IM} \left( \pi, \nu \right)$-setting, $M$ in fact satisfies a WPI w.r.t. $\bar{\beta}$. 
Thus while our hybrid analysis is not sharp in this particular example, which is not to be expected given its generality, it is quantitatively not too far off. The loss here could due to a possibly suboptimal estimate of Proposition~\ref{prop: Exp_USS_spec_gap}
regarding the ideal slice sampler.

\end{remark}

\subsection{Stepping-out and shrinkage on bimodal targets}

Given the general difficulty of sampling the uniform distribution on a slice already in a univariate setting, Neal proposed in \cite{Neal2003Slice} a stepping-out-shrinkage procedure to construct Hybrid Slice Samplers that mimic an ideal uniform one. 

We consider 
 $\sfG\subseteq\mathbb{R}$ equipped with $\mathcal{G}=\mathcal{B}(\sfG)$ and $\nu$ being the $1$-dimensional Lebesgue measure. Moreover, we introduce a special type of bimodality that was previously considered in \cite[Section~4.1]{Latuszynski2014convergence}.

\begin{definition}\label{def:t_bimodal}
Let $\varpi: \sfG \to (0,\infty)$ 
and let $t_1,t_2\in (0,\Vert \varpi \Vert_{\infty})$ satisfying $t_{1} \leq t_{2}$. We say that $\varpi$ is \emph{$\left( t_{1}, t_{2} \right)$-bimodal} if 
\begin{itemize}
    \item for all $t \in \mathsf{T} \setminus \left[ t_{1}, t_{2} \right)$, the superlevel set $\mathsf{G} \left( t \right)$ (of $\varpi$) consists of a single interval, and
    \item for all $t \in \left[ t_{1}, t_{2} \right)$, the superlevel set $\mathsf{G} \left( t \right)$ consists of a pair of disjoint sub-intervals $\mathsf{G} \left( t \right) = \mathsf{G}_{1} \left( t \right) \sqcup \mathsf{G}_{2} \left( t \right)$ such that commonly-labelled sub-intervals are nested, i.e. for $i = 1, 2$ and $t_{1} \leq s \leq t <t_{2}$, there holds the inclusion $\mathsf{G}_{i} \left( t \right) \subseteq \mathsf{G}_{i} \left( s \right)$.
\end{itemize}
Moreover, given a $(t_1,t_2)$-bimodal density $\varpi$, define $\delta_{\varpi}:(0,\Vert \varpi \Vert_\infty )\to\left[0,\infty\right)$ by
\begin{align*}
    \delta_{\varpi}(t)= 
    \begin{cases}
        \mathsf{dist} \left(\mathsf{G}_{1} (t), \mathsf{G}_{2} (t) \right) & t \in \left[t_{1},t_{2}\right)\\
        0 & \text{otherwise},
\end{cases}
\end{align*}
where $\mathsf{dist} \left( \mathsf{G}_{1} \left(t\right), \mathsf{G}_{2} \left(t\right) \right) := \inf \left\{ \left|x-y\right| : x \in \mathsf{G}_{1} \left( t \right), y \in \mathsf{G}_{2} \left( t \right) \right\}$. Furthermore set $\Delta_{\varpi} := \sup\left\{ \delta_{\varpi}\left(t\right) : t \in (0,\Vert \varpi \Vert_\infty )\right\}$ and $m_\varpi := \inf \{ \nu ( \sfG(t) ) : t\in\left[t_{1},t_{2}\right) \}$ with the convention that the infimum over the empty set is $\infty$.
\end{definition}
Note that, any \textit{unimodal} density $\varpi$ is $(t,t)$-bimodal for any $t\in(0,\Vert \varpi \Vert_{\infty})$ with $\Delta_{\varpi}=0$ and $m_\varpi=\infty$. For
any \textit{bimodal} density $\varpi$ with two global maxima areas the number $\Delta_\varpi$ denotes the distance between them and in the case of two maxima points $m_\varpi=0$ with the modes being separated by $\Delta_\varpi$. If there are not two global maxima, then the modes are at least $\Delta_\varpi$ far apart from each other and $m_\varpi>0$.

For $\varpi \colon \sfG \to (0,\infty)$ and
parameter $h>0$ we consider the stepping-out shrinkage Hybrid Slice Sampler introduced in \cite{Neal2003Slice}. For a schematic view of the transition mechanism we refer to Figure~\ref{fig:step_slice} and for a pseudo-code description check \cite[Algorithm~1]{Latuszynski2014convergence}. We emphasize that the mechanism/algorithm only requires as input the parameter $h$ and an oracle that provides function evaluations of $\varpi$ upon request.

\begin{figure}
    \centering
    \includegraphics[width=0.925\linewidth]{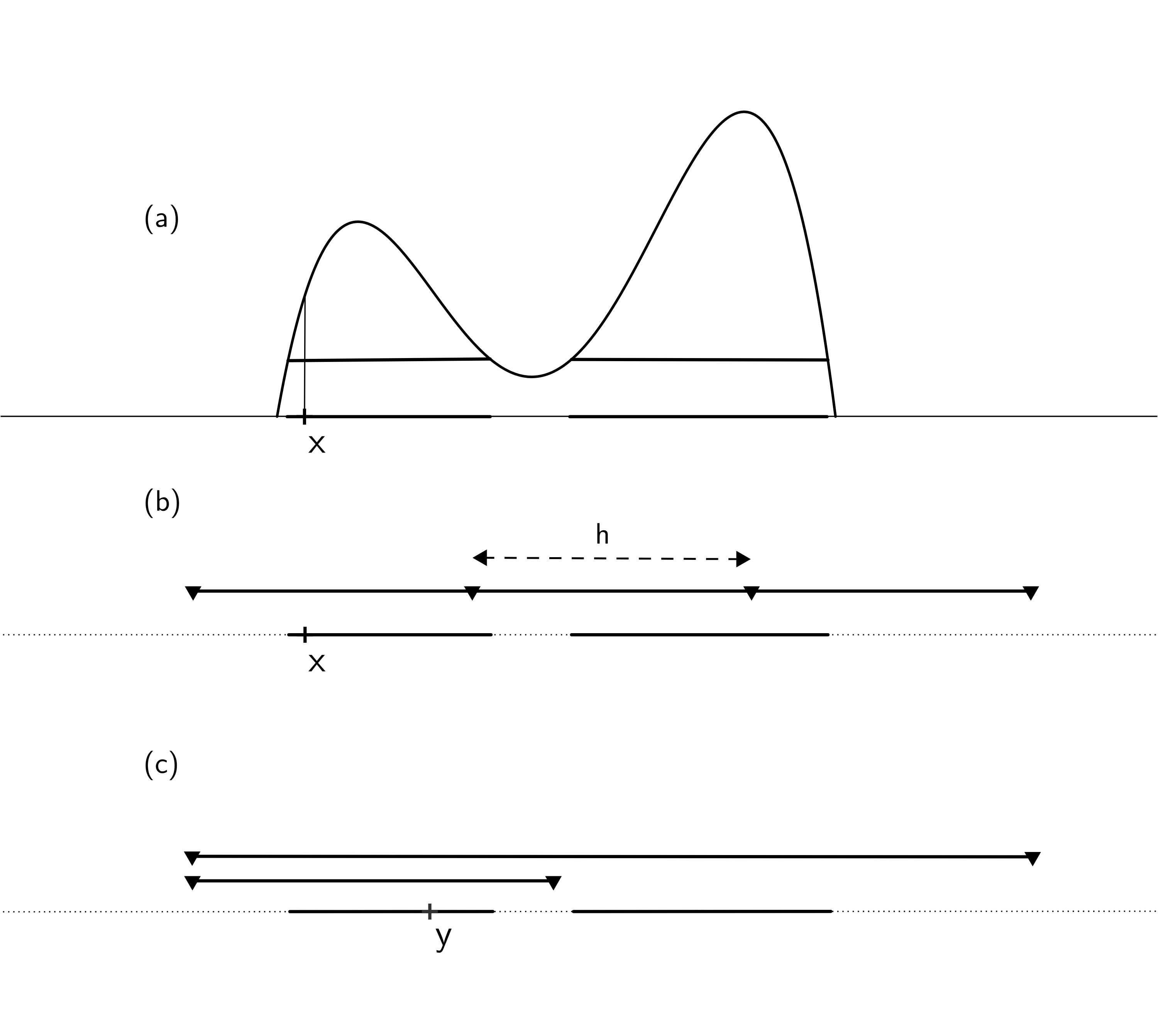}
    \caption{A single transition of the 
    stepping-out and shrinkage Hybrid Slice Sampler from $x$ to $y$ is schematically presented:
    (a) A slice is chosen uniformly distributed from $(0,\varpi(x))$, 
    see horizontal bold line for the slice.
    (b) An interval of width $h$ is randomly positioned around $x$ and
    expanded in steps of size $h$ until both ends are outside the slice.
    (c) The next state $y$ is found by sampling uniformly within the interval 
    until a state is inside the slice. Points picked that are outside the
    slice are used to shrink the interval.}
    \label{fig:step_slice}
\end{figure}

Fix $h\geq\Delta_\varpi$ and
assume that there are $t_1,t_2\in (0,\Vert \varpi \Vert_{\infty})$, such that $\varpi$ is $(t_1,t_2)$-bimodal.
Then, the algorithmic description yields for $t\in(0,\Vert \varpi\Vert_\infty)$
the stepping-out shrinkage Markov kernel $H_t$ on $\sfG(t)$, check \cite[Lemma~14]{Latuszynski2014convergence}, which is for $x\in \sfG(t)$ and $A\in\mathcal{B}(\sfG(t))$ given by
\begin{align*}
    H_t \left( x, A \right) = \lambda\left(t\right) \cdot \nu_t(A) + \left( 1 - \lambda \left(t\right) \right) \cdot \left[ \mathbf{1}_{\mathsf{G}_1(t)}(x) \cdot \nu_{t,1}(A) + \mathbf{1}_{\mathsf{G}_2(t)}(x) \cdot\nu_{t,2}(A)\right],
\end{align*}
where
\begin{align*}
    \lambda\left(t\right)	&:=
    \begin{cases}
        \frac{h-\delta_{\varpi}\left(t\right)}{h}\cdot\frac{\nu\left(\sfG(t)\right)}{\nu\left(\sfG(t)\right)+\delta_{\varpi}\left(t\right)} & t\in\left[t_{1},t_{2}\right);\\
        1 & \text{otherwise};
    \end{cases} \\
    \nu_{t,i}\left(A\right)	&:=
    \begin{cases}
        \frac{\nu\left(\mathsf{G}_{i}\left(t\right)\cap A\right)}{\nu\left(\mathsf{G}_{i}\left(t\right)\right)} & t\in\left[t_{1},t_{2}\right);\\
        0 & \text{otherwise};
    \end{cases}\quad\text{for }i=1,2.
\end{align*}
Observe that for $t\in [t_1,t_2)$, the measure $\nu_{t,i}$ coincides with the uniform distribution in $\sfG_i(t)$, and for $t \in (0\Vert \varpi \Vert_\infty)\setminus [t_1,t_2)$, it holds that $H_t = \nu_t$, i.e. one samples exactly from the uniform distribution on the superlevel set. 
We now summarize properties of  $H_{t}$ which are proven in \cite[Lemma~9]{Latuszynski2014convergence}.

\begin{proposition}
For $\varpi\colon \sfG \to (0,\infty)$ assume that there are $t_1,t_2\in (0,\Vert \varpi \Vert_{\infty})$, such that $\varpi$ is $(t_1,t_2)$-bimodal and let $h\geq \Delta_\varpi$. Then, for any $t\in(0,\Vert \varpi\Vert_\infty)$ the Markov kernel $H_t$ is $\nu_t$-reversible and positive (on $\ELL(\nu_t)$). Moreover, for any $f\in \ELL_0(\nu_t)$ we have
    \begin{align*}
        \mathcal{E}_{\nu_{t}} \left(H_{t},f\right) \geq \lambda\left(t\right) \Vert{f\Vert}_{\nu_{t}}^2.
    \end{align*}
    In particular, within the $\mathsf{HybridSlice} \left( \pi, \nu, \left(H_{t} \right) \right)$-setting  Assumption~\ref{assms} is satisfied with $\widetilde{\beta} (s,t)= 1/4 \cdot \mathbf{1}_{(0,\lambda (t)^{-1})} (s).$
    \end{proposition}

As a consequence of our main theorem we  obtain the following comparison result between Ideal and stepping-out shrinkage Hybrid Slice Sampling.

\begin{corollary}
    For $\varpi\colon \sfG \to (0,\infty)$ and $t_1,t_2\in (0,\Vert \varpi \Vert_{\infty})$ let $\varpi$ be $(t_1,t_2)$-bimodal and set $h\geq \Delta_\varpi$. Then, for $H$
    from the $\mathsf{HybridSlice} \left( \pi, \nu, \left(H_{t} \right) \right)$-
    and $U$ from the $\mathsf{Slice} \left(\pi,\nu\right)$-setting for all $f\in \ELL_0(\pi)$ we have 
    \begin{align*}
        \rho \, \mathcal{E}_\pi\left(U,f\right) \leq \mathcal{E}_\pi\left(H,f\right) \leq \mathcal{E}_\pi\left(U,f\right),
    \end{align*}
    where
    \begin{align*}
        \rho := \frac{h-\Delta_\varpi}{h} \cdot \frac{m_{\varpi}}{m_{\varpi}+\Delta_\varpi}. 
    \end{align*}
\end{corollary}

\begin{proof}
    For $t \in \left[t_{1},t_{2}\right)$, observe that i) $\lambda \left(t\right) = \frac{h-\delta_\varpi\left(t\right)}{h} \cdot \frac{\nu(\sfG(t))}{\nu(\sfG(t) + \delta_\varpi\left(t\right)}$, which is a decreasing function of $\delta_\varpi\left(t\right)$ and an increasing function of $\nu(\sfG(t))$, and ii) that $\delta_\varpi\left(t\right)\leq\Delta_\varpi$ and $\nu(\sfG(t))\geq m_{\varpi}$ by construction, so that $\lambda \left(t\right) \geq\rho$. For $t\in(0,\Vert \varpi \Vert_\infty) \setminus \left[t_{1},t_{2}\right)$, we even have $\lambda\left(t\right) = 1 \geq \rho$, and so that $\lambda\left(t\right) \geq \rho$ for all $t \in (0,\Vert \varpi \Vert_\infty)$. The first inequality then follows by Corollary~\ref{cor:Ht_gap_comp}; the second follows as a direct consequence of Theorem~\ref{thm:SS_comp}.
\end{proof}

The latter result implies a bound on
the relative spectral gap of the hybrid sampler, compared to its ideal counterpart, thereby quantifying the degradation in performance which is incurred by the use of the stepping-out and shrinkage procedure. Moreover, it provides a quantitative version of \cite[Corollary~2]{Latuszynski2014convergence}. We further observe:

\begin{enumerate}
    \item For fixed $\eta_1\in(0,1)$, by taking $h$ sufficiently large, one can always ensure $\rho \geq(1-\eta_1) \frac{m_{\varpi}}{m_{\varpi}+\Delta_{\varpi}}$. For $m_{\varpi}>0$ this quantity is small when $\Delta_{\varpi}$ compared to $m_{\varpi}$ is large, i.e. when the maximal distance between the two modes is large relative to the Lebesue mass around the more peaked mode.
    \item By taking $h = \left(1+\eta_2\right) \cdot \Delta_\varpi$ for some $\eta_2 > 0$, one can ensure that $\rho \ge \frac{\eta_2}{1+\eta_2} \cdot \frac{m_{\varpi}}{m_{\varpi} + \Delta_\varpi}$. Intuitively, by taking $h = 2 \Delta_\varpi $, the performance guarantee is only worse than that of $h = \infty$ by a factor of $2$.
 \end{enumerate}

\subsection{Hit-and-Run within slice sampling}
\label{sec: HaR_within_ss}

Consider $(\sfG,\mathcal{G},\nu) = (\mathbb{R}^d,\mathcal{B}(\mathbb{R}^d),\mathrm{Leb})$, such that $\varpi \colon \sfG \to (0,\infty)$ is the not necessarily normalized Lebesgue density of the target measure $\pi$ which satisfies $\Vert \varpi \Vert_\infty = \mathrm{ess}_\nu \sup \varpi$.
In this setting for any $t\in (0,\Vert \varpi \Vert_\infty)$ the probability measure $\nu_t$ fulfills 
\[
\nu_t(A) = \frac{\mathrm{Leb}(A\cap \sfG(t))}{\mathrm{Leb}(\sfG(t))}, \quad A\in\mathcal{B}(\sfG(t)).
\]
Hence it coincides with the uniform distribution on $\sfG(t)$, cf. also Algorithm~\ref{alg:SSS}. Therefore, Ideal Slice Sampling requires sampling of uniform distributions in $\mathbb{R}^d$ and an appropriate Hybrid Slice Sampler can be based on the classical \textit{Hit-and-Run} algorithm, cf. \cite{Lovasz2004hitandrun,Smith1984efficient}, that offers approximate Markov chain sampling of $\nu_t$.

We begin by describing Hit-and-Run on a generic convex body $K\subset \mathbb{R}^d$, namely a convex bounded set with non-empty interior. 
Let $B(x,r):=\{z\in\mathbb{R}^d\colon \vert z\vert \leq r\}$ be the Euclidean ball of radius $r>0$ around $x\in\mathbb{R}^d$, where $\vert \cdot \vert$ denotes the Euclidean norm.
Let $S_{d-1}:= \partial B(0,1)$ be the unit sphere in $\Rd$, that is, the boundary of the ball around $0$ with radius $1$, and denote by $\mathcal{U}_{d-1}$ the uniform distribution on $S_{d-1}$.
For $x\in K, \theta\in S_{d-1}$ let $\ell_K(x,\theta) := \{ s\in\mathbb{R} \colon x+s\theta \in K \}$, so that the Hit-and-Run Markov kernel can be represented as 
\[
    \mathcal{H}_K(x,A) 
    := \int_{S_{d-1}} \int_{\ell_K(x,\theta)} 
   \frac{\mathbf{1}_A(x+s\theta)}{\mathrm{Leb}(\ell_K(x,\theta))}\dif s\; \mathcal{U}_{d-1}(\dif \theta).
\]
The corresponding transition mechanism first samples a direction $\theta$ uniformly distributed from $S_{d-1}$, and then samples a point uniformly distributed on the intersection of $K$ and the line passing through $x$ in direction $\theta$. It
is well-known that $\mathcal{H}_K$ is reversible w.r.t. the uniform distribution on $K$, say $\mu_K$, and in \cite[Section~3.1]{Rudolf2013positivity} it has been established that it is positive on $\ELL(\mu_K)$.

Moreover, an immediate consequence of \cite{Lovasz2004hitandrun} (that has essentially also been used in \cite[Section~4.2]{Ru12}), via Cheeger's inequality \cite{lawler1988bounds} and the aformentioned positivity, yields the following statement about the spectral gap of $\mathcal{H}_K$.
\begin{theorem}\label{thm:lv-handr-conductance}
    Let $K\subseteq \Rd$ be a convex body with
    \begin{align*}
    r_K &:= \sup \left\{ r > 0 : \exists x \in K\, \text{ such that } \, B(x, r) \subseteq K \right\}, \\
    R_K &:= \inf \left\{ R > 0 : \exists x \in K \, \text{ such that } \, B(x, R) \supseteq K \right\}.
\end{align*}
With the condition number $\kappa_K:= \frac{R_K}{r_K} \in [1, \infty)$ we have for all $f\in\ELL_0(\mu_K)$ that
    \begin{align*}
       \mathcal{E}_{\mu_K}(\mathcal{H}_K,f)& \geq 2^{-33}
        d^{-2} \kappa_K^{-2} \Vert f \Vert_{\mu_K}^2.
    \end{align*}
\end{theorem}

\noindent
For meaningful statements regarding Hit-and-Run \textit{within} slice sampling, i.e. for employing Hit-and-Run on the level sets, we require further terminology regarding $\varpi$. 
\begin{definition}  \label{def: potential}
 Define the \textit{potential} (w.r.t. $\pi$ or $\varpi$) as
$V(x):=-\log \varpi(x)$ for all $x\in\mathbb{R}^d$ and note that
$\pi(\dif x) \propto \varpi(x) \dif x = \exp(-V(x)) \dif x$. We say that $\pi$ is \textit{log-concave} if the potential $V\colon \mathbb{R}^d \to \mathbb{R}$ is convex. Moreover, we call $\pi$ \textit{quasi-log-concave} if $V$ is quasi-convex, i.e. for any
$t\in \mathbb{R}$ the set
\[
\{z\in\mathbb{R}^d\colon V(z)< t\}
\]
is convex, that is, all sub-level sets of $V$ are convex.
\end{definition}
Observe, that if $\pi$ is quasi-log-concave, then for any $t\in
\sfT
$ the level set 
$\sfG(t)=\{x\in\mathbb{R}^d\colon \varpi(x)>t\}$ is also convex. Now we are able to state our first comparison result, which is a consequence of Theorem~\ref{thm:lv-handr-conductance} and Corollary~\ref{cor:Ht_gap_comp}.

\begin{corollary}\label{cor:qlc-hss-handr} 
    Let $\pi$ be quasi-log-concave and for any $t\in
    \sfT
    $ set 
    \begin{align*}
    r(t) &:= \sup \left\{ r > 0 \colon \exists x \in \sfG(t) \, \text{ such that } \, B(x, r) \subseteq \sfG(t) \right\}, \\
    R(t) &:= \inf \left\{ R > 0 \colon \exists x \in \sfG(t) \, \text{ such that } \, B(x, R) \supseteq \sfG(t) \right\}.
\end{align*}
Define the \emph{conditioning profile} $\kappa_{\varpi} ( t ):= \frac{R(t)}{r(t)}$. Let $U$ be specified by the setting $\mathsf{Slice} \left( \pi, \mathrm{Leb}\right)$ and the Hit-and-Run within slice sampling Markov kernel $H$ be determined by the setting $\mathsf{HybridSlice} ( \pi, \nu, (H_t))$ with $H_t = \mathcal{H}_{\sfG(t)}$ for any $t\in \sfT
$, i.e. a Hit-and-Run transition on $\sfG(t)$ is performed. Then, for all $s>0$ and $f\in \ELL(\pi)$ we have
 \begin{align}
 \label{al: Hit_and_run_SS}
    \calE_\pi (H,f) \le \calE_\pi (U,f)\leq s\, \calE_\pi (H,f) + \frac{\|f\|^2_{\osc,\pi}}{4c}   \int_{\sfT} \mathbf{1}_{(0,2^{33}\, d^{2} \kappa_{\varpi}(t)^{2})}(s) \, m \left( t \right) \dif t.
\end{align}
    In particular, if $\sup_{t \in \sfT} \kappa_{\varpi} ( t ) = \overline{\kappa}_\varpi < \infty$, then 
    \begin{align*}
        \mathcal{E}_\pi \left( U, f \right) & \leq 2^{33} \,d^2  \overline{\kappa}_\varpi^2 \, \mathcal{E}_\pi \left( H, f \right).
    \end{align*}
\end{corollary}

\begin{proof}
For any $t\in\sfT$, by the quasi-log-concavity, $\sfG(t)$ is convex, such that by the fact that $\mathrm{Leb}(\sfG(t))\in(0,\infty)$ we have that $\sfG(t)$ is a convex body (bounded, convex set with non-empty interior). 
    By combining Theorem \ref{thm:lv-handr-conductance} with Remark~\ref{rem: SPI}, for any $t\in\sfT$ we have
    \[
        \Vert f \Vert_{\nu_t}^2 \leq s \cdot \calE_{\nu_t} \left( H_t, f \right) +\mathbf{1}_{(0,2^{33}\, d^{2} \kappa_{\varpi}(t)^{2})}(s)  \cdot \|f\|_{\osc,\nu_t}^2.
    \]    
Now, the final statement follows by Corollary~\ref{cor:Ht_gap_comp}.
\end{proof}
We apply this result to well-conditioned log-concave targets that attracted intense theoretical interest (see e.g. \cite{chewi2023log}), due in part to its relative tractability.
\begin{example}
    Assume that the potential $V$ of a log-concave $\pi$ is $m$-strongly-convex and $L$-smooth for some 
    $0 < m  < L$, that is, for all $x, y \in \mathbb{R}^d$ we have
\begin{equation}\label{eq:sandwich}
    \frac{m}{2} \, \vert y - x  \vert^2  \leq V (y) - V(x) - \langle \nabla V (x), y - x \rangle   \leq \frac{L}{2} \,\vert y - x \vert^2,
\end{equation}
where $\langle \cdot, \cdot \rangle$ denotes the Euclidean inner-product. By convexity of $V$, it is immediate that we are in the setting of Corollary \ref{cor:qlc-hss-handr}. Without loss of generality, assume that $V$ is uniquely minimised at $x = 0$, and that $V \left( 0 \right) = 0$. Applying \eqref{eq:sandwich} with one argument equal to $0$ establishes for all $z\in\sfG=\mathbb{R}^d$ that
    \begin{align*}
        \frac{m}{2} \, \vert z\vert^{2} 
        \leq V(z) \leq \frac{L}{2} \, \vert z\vert^{2}.
    \end{align*}
    By the former inequalities and 
    $
    \sfG(t) = \{ z\in\mathbb{R}^d \colon V(z) <\exp(-t) \}
    $
    we have 
    \[
    B(0,\sqrt{2\exp(-t)}/\sqrt{L}) \subseteq \sfG(t) \subseteq B(0,\sqrt{2\exp(-t)}/\sqrt{m}),
    \]
    such that $\overline{\kappa}_\varpi \leq \sqrt{m/L}$ and by Corollary~\ref{cor:qlc-hss-handr} we obtain
    \begin{align}
    \label{al: L_smooth_etc}
        \mathcal{E}_\pi \left( U, f \right) & \leq 2^{33} \cdot  \frac{m\,d^2}{L} \cdot \mathcal{E}_\pi \left( H, f \right).
    \end{align}
\end{example}

By results of \cite{Natarovskii2021quantitative}, one expects for spherically-symmetric log-concave $\varpi$ in dimension $d$ that $U$ admits a spectral gap of order $d^{-1}$. Inequality \eqref{al: L_smooth_etc} thus suggests that Hit-and-Run within slice sampling $H$ admits a spectral gap of order $d^{-3}$; this is worse by a factor $d^2 $ than the best known result for the Random Walk Metropolis algorithm as applied to the same problem; see \cite{Andrieu2022rwm} for details.

A key feature of ideal slice sampling is its robustness in heavy-tailed scenarios. While performance may degrade somewhat in this setting (see e.g. \cite{Schar2023revisited}), it has been observed that $U$ still exhibits a spectral gap; this stands in contrast to, for instance, the Random Walk Metropolis algorithm \cite{ Andrieu2023wpi,Jarner2000}. Bearing this in mind, we may expect that for heavy-tailed $\varpi$ with well-conditioned superlevel sets, Hit-and-Run within slice sampling as practically viable option. The following example illustrates this possibility in a specific case.

\begin{example}
    For $d \in \mathbb{N}, m >2$ let $\varpi(x) = ( 1 + \vert x \vert^2 )^{-(d + m)/2}$ for $x\in\mathbb{R}^d$, i.e. $\pi$ corresponds to a $d$-variate Student-t-type distribution with $m$ degrees of freedom.
    Note that the superlevel sets are simply given by Euclidean balls, 
    such that $\kappa_\varpi(t)=1$ for any $t\in\sfT$.
    Moreover, \cite[Theorem~2.4]{Schar2023revisited} proves $\frac{m - 1}{(d + 1)(d + m - 1)}$ as a lower bound of the spectral gap of $U$. Then, in combination with Corollary~\ref{cor:qlc-hss-handr}, for any $f\in\ELL(\pi)$ this yields
    \begin{align*}
       \Vert f\Vert_\pi^2 & \leq 2^{33} \cdot \frac{d^2 (d + 1)(d + m - 1)}{m - 1}\,  \mathcal{E}_\pi ( H, f ).
    \end{align*}
\end{example}

The examples considering so far have conditioning profile that is uniformly bounded from above. For distributions $\pi$ whose tail behaviour is heterogeneous in different directions, the conditioning profile may degenerate in certain limits, see Figure~\ref{fig:poorly-cond}. In such a setting, one expects that Hit-and-Run within slice sampling will degrade. Therefore, it is not surprising that in the following example, based on results of Appendix~\ref{sec:degenerating}, our theory leads to a comparison result that does not imply a lower bound of the spectral gap of the Hit-and-Run within slice sampling from the Markov kernel $U$ of the setting $\mathsf{Slice} \left( \pi, \mathrm{Leb}\right)$. 
\begin{figure}
    \centering
    \includegraphics[width=0.47\linewidth]{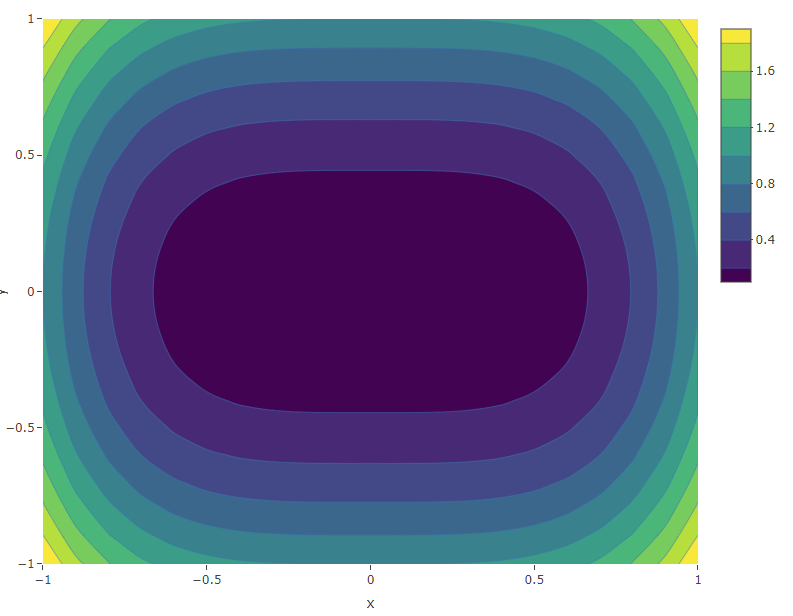}
    \includegraphics[width=0.47\linewidth]{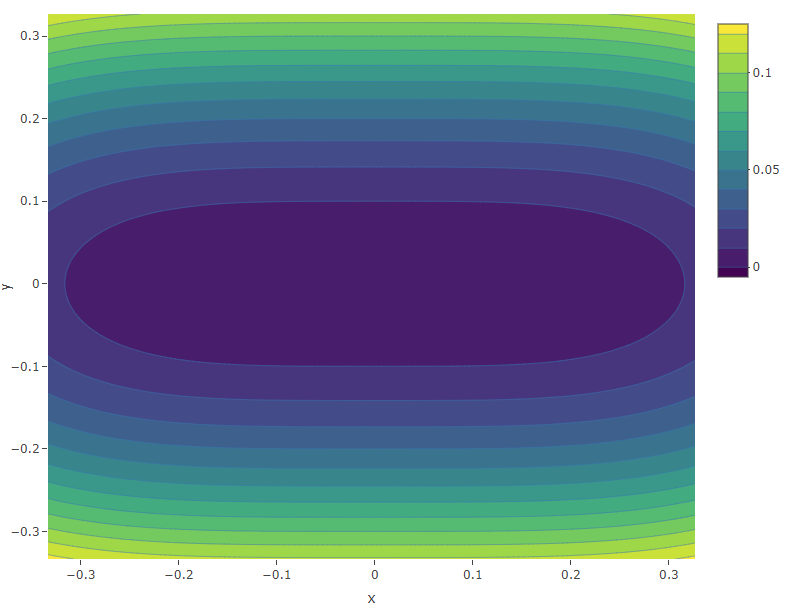}
    \caption{Contour lines of condition profile determined by the potential $V(x,y)=x^4+y^2$ for $x,y\in\mathbb{R}$.}
    \label{fig:poorly-cond}
\end{figure}

\begin{example}
Fix $I\subseteq \{1,\dots,d\}$ 
and let $I^c:=\{1,\dots,d\}\setminus I$. We consider a specific potential $V\colon \mathbb{R}^d \to [0,\infty)$ that determines $\pi$. For $x=(x_1,\dots,x_d)\in\mathbb{R}^d$ we set
\[
V(x) = \sum_{i\in I} x_i^2 + \sum_{i\in I^c} x_i^4.
\]
Note that many of the following arguments rely on the technical developments of Appendix~\ref{sec:degenerating}, in particular Example~\ref{ex: appendix}.
With $\omega_d:={\rm Leb} (B(0,1))$ observe that we have
\begin{align*}
c & = \int_{\mathbb{R}^d} \exp(-V(x)) {\rm d} x 
   \underset{\eqref{eq: V_small}}{\geq} \int_{B(0,1)} \mathrm{e}^{-\vert x \vert^2} {\rm d}x
  \geq \mathrm{e}^{-1} \omega_d,
\end{align*}
where the last inequality follows by the fact that $\inf_{x\in B(0,1)} \mathrm{e}^{-\vert x \vert^2} = \mathrm{e}^{-1}$.
For $d\geq 4$,  by \eqref{eq: main_ex_app}, by the former estimate and \eqref{al: Hit_and_run_SS} from Corollary~\ref{cor:qlc-hss-handr} we obtain for $H$, given by Hit-and-Run within slice sampling, compared to $U$ that for any $s>0$ and all $f\in \ELL_0(\pi)$,
\begin{align*}
        \calE_\pi (U,f) & \leq s\, \calE_\pi (H,f) + \|f\|^2_{\osc,\pi} \, d^{2d}\, 2^{6d}\, s^{-d/8-1/4}.
\end{align*}
Because of the exponential dependence on the dimension, the latter bound appears on the first glance not very useful. However, observe that in case of a lower bound of the spectral gap, i.e.  $\gamma_d \cdot \Vert f \Vert_\pi^2 \leq \calE_\pi (U,f)$, for some $\gamma_d>0$ and all $f\in \ELL_0(\pi)$, we have that $H$ satisfies by Corollary~\ref{cor: spec_gap_U_implies_WPI_H} a WPI w.r.t.
\[
    \bar{\beta}(s) := \frac{d^{2d}\, 2^{6d}}{\gamma_d^{d/8+5/4}} \; s^{-d/8-1/4}.
\]
By Corollary~\ref{cor: MC_implication} and the first bullet point of Remark~\ref{rem: alpha_n} this yields for a Markov chain $(X_n)_{n\in\mathbb{N}_0}$ with Markov kernel $H$ and constant $M:= \Big\Vert \frac{{\rm d} \mathbb{P}^{X_0}}{{\rm d}\pi} - 1 \Big\Vert_{\osc,\pi}^{1/2} $ that
\[
\sup_{A\in \mathcal{B}(\mathbb{R}^d)} \big\vert \mathbb{P}^{X_n}(A) - \pi(A) \big\vert
\leq M\; d^d 2^{3d} \Big(\frac{d+10}{\gamma_d}\Big)^{\frac{d+10}{16}}\cdot n^{-\frac{d+2}{16}}, 
\qquad n\in\mathbb{N}.
\]
For given $\varepsilon\in(0,1)$, to achieve that the left hand-side of the last inequality is smaller than $\varepsilon$, it is sufficient to choose $n\in\mathbb{N}$ so that
\[
n> 2^{48}\,(M \varepsilon)^{16/(d+2)}\, d^{16} \,\Big(\frac{d+10}{8\gamma_d}\Big)^3.
\]
Apart of the large absolute constant, the dependence w.r.t. the dimension may be of interest. Observe, that for fixed $\varepsilon$, $\gamma_d$ behaving polynomially in $d^{-1}$ and $M\leq 2^{(d+2)}$, the number $n$ increases only polynomially in $d$. This shows that in such scenarios, there is no curse of dimensionality in terms of an exponential dependence on $d$ for the number of iterations required to achieve an error less than $\varepsilon$.
\end{example}

We add a discussion of further approaches on the sampling problem considered in this section that seem particularly relevant. In specific scenarios, other on-slice kernels (besides Hit-and-Run) are of substantial interest. For superlevel sets with some coordinate-related structure, Gibbs sampling (also `Coordinate Hit-and-Run') is a natural option, for which results have recently been obtained in the convex setting, see \cite{laddha2023convergence}. For superlevel sets which take the form of polytopes, yet more on-slice kernels are available, including the Dikin walk \cite{kannan2009random}, the Vaidya and John walks \cite{chen2018fast}, and various gradient-based samplers which introduce additional geometric structure to the problem (e.g. \cite{lee2022geodesic}). Some of these methods come with theoretical guarantees in the form of estimates on the conductance or spectral gap, which can readily be used in our framework. 
Another aspect of our analysis is that we assume the geometry of a superlevel set to be given. In other words, the method does not attempt to improve its conditioning profile prior to applying the on-slice kernels. Moving beyond this amounts to incorporate some form of `on-slice preconditioning', which could turn out to be valuable. As a hint in this direction, for a convex body $K$, the Loewner-John Ellipsoid Theorem (see e.g. \cite{john2014extremum}) implies the existence of a linear change of variables which guarantees that $\kappa_K  \leq d$ (cf. the notation of Theorem~\ref{thm:lv-handr-conductance}). 
The extent to which such a strategy can be implemented effectively in practice remains to be seen.

\section{Acknowledgements}

SP's research was supported by EPSRC grant Bayes4Health, `New Approaches to Bayesian Data Science: Tackling Challenges from the Health Sciences' (EP/R018561/1). Part of AQW's research was supported by EPSRC grant `CoSInES (COmputational Statistical INference for Engineering and Security)' (EP/R034710/1). 
DR and BS were supported by the German Research Foundation (DFG grant number 522337282).

\appendix

\section{Exponential distribution Ideal Slice Sampling}
\label{app:exp_IMH}

We consider $\sfG=[0,\infty)$, equipped with the corresponding Borel $\sigma$-algebra $\mathcal{G}$ and exponential distribution $\nu$ with parameter $\lambda>0$, i.e. $\nu(\dif x)=\lambda e^{-\lambda x} \, \dif x$. Moreover, let $\pi$ also be an Exponential distribution with parameter $\alpha>0$ on $(\sfG,\mathcal{G})$. Then, $\pi(\dif x) = \alpha e^{-\alpha x}\, \dif x = \frac{\alpha}{\lambda} e^{-(\alpha-\lambda)x} \nu(\dif x)$, such that $\varpi(x)=e^{-(\alpha-\lambda)x}$. We consider the 
$\mathsf{Slice} \left( \pi, \nu \right)$-setting with corresponding Markov kernel
\[
U(x,A) 
= \frac{1}{\varpi(x)} \int_0^{\varpi(x)} \frac{\nu(A\cap\sfG(t))}{m(t))} \dif t.
\]
Set $r:=\alpha-\lambda$ and note that for $t\in\sfT$ we have
\[
\sfG(t) = \begin{cases}
            [0,-\frac{\log(t)}{r}), & r>0,\\
            [0,\infty), & r=0,\\
            [-\frac{\log(\max\{1,t\})}{r},\infty), & r<0,
          \end{cases}   
\quad \text{and} \quad
m(t)=\begin{cases}
            1-t^{\lambda/r}, & r>0,\\
            1, & r=0,\\
            (\max\{1,t\})^{\lambda/r}, & r<0.
     \end{cases}  
\]
Observe that the case $\alpha=\lambda$ is trivial, since then $\pi$ is sampled directly.

Given a bijective map $T\colon \sfG\to \widetilde{\sfG}$ and corresponding measurable space $(\widetilde{\sfG},\widetilde{\mathcal{G}})=(T(G),T(\mathcal{G}))$, define the pushforward measure of $\pi$ (w.r.t. $T$) on $(\widetilde{\sfG},\widetilde{\mathcal{G}})$ by $\widetilde{\pi}(B) = \pi(T^{-1}(B))$, with $B\in\widetilde{\mathcal{G}}$, and the pushforward Markov kernel $\widetilde{U} \colon \widetilde{\sfG}\times \widetilde{\mathcal{G}}\to [0,1]$ by
\[
\widetilde{U}(z,B) := U(T^{-1}z,T^{-1}(B)), \qquad (z,B)\in \widetilde{\sfG}\times\widetilde{\mathcal{G}}.
\]
By the fact that $U$ is $\pi$-reversible and positive also $\widetilde{U}$ is $\widetilde{\pi}$-reversible and positive, see \cite[Proposition~20 and 21]{RudolfSprungk2022}.
Moreover, by \cite[Theorem~27]{RudolfSprungk2022} it follows
\begin{equation}
\label{eq: spec_gap_transf}
 \inf_{f\in \ELL_0(\pi),\, \Vert f \Vert_\pi=1} 
    \mathcal{E}_\pi(U,f) = \inf_{g\in \ELL_0(\widetilde{\pi}),\, \Vert g \Vert_{\widetilde{\pi}}=1} 
    \mathcal{E}_{\widetilde{\pi}}(\widetilde{U},g),
\end{equation}
which means that the spectral gaps of $U$ and $\widetilde{U}$ coincide, cf. Remark~\ref{rem: SPI}. 

Therefore, by using a suitable $T$, coupling arguments and \cite[Proposition~30]{ollivier2009ricci} we show a lower bound of the right hand-side of \eqref{eq: spec_gap_transf} and conclude a lower bound of the spectral gap of $U$ for exponential distribution Ideal Slice Sampling. 

\begin{proposition}
\label{prop: Exp_USS_spec_gap}
    In the setting of this section, we have
    \[
        \inf_{f\in \ELL_0(\pi),\, \Vert f \Vert_\pi=1} 
    \mathcal{E}_\pi(U,f) \geq \gamma(\alpha,\lambda):= 
    \begin{cases}
        \frac{\alpha+\lambda}{2\alpha}    & \alpha\geq \lambda,\\
        \big(\frac{\alpha}{2\lambda-\alpha}\big)^2 & \alpha<\lambda.
    \end{cases}
    \]
\end{proposition}
\begin{proof}
We distinguish two cases:

\textbf{1st case:} 
We assume that $r=\alpha-\lambda>0$, which may be interpreted as heavy-tailed reference measure setting. Consider the bijection $T\colon [0,\infty) \to (0,1]$ given by $T(x)=e^{-\lambda x}$ with $T^{-1}(x)=-\log(x)/\lambda$. For $(z,B)\in (0,1]\times \mathcal{B}((0,1])$ with $\mathrm{Leb}$ being the 1-dimensional Lebesgue measure, we have
\begin{align*}
\widetilde{U}(z,B) & = \frac{\lambda}{z^{r/\lambda}} \int_0^{z^{r/\lambda}} \frac{1}{(1-t^{\lambda/r})} \int_{T^{-1}(B)\cap [0,-\log(t)/r)} e^{-\lambda y} \dif y \dif t \\
& = \frac{\lambda}{z^{r/\lambda}} \int_0^{z^{r/\lambda}} \frac{1}{(1-t^{\lambda/r})} \int_{T^{-1}(B\cap (t^{\lambda/r},1])} e^{-\lambda y} \dif y \dif t \\
&= \frac{1}{z^{r/\lambda}} \int_0^{z^{r/\lambda}} 
\frac{\mathrm{Leb}(B\cap (t^{\lambda/r},1])}{(1-t^{\lambda/r})} \dif t,
\end{align*}
where in the last equality we applied integration by substitution w.r.t. $T^{-1}$. The Markov kernel $\widetilde{U}$ corresponds to Ideal Slice Sampling targeting $\widetilde{\pi}(\dif z) \propto z^{r/\lambda} \dif z.$
For arbitrary $z_1,z_2\in(0,1]$ observe that the probability measure $\xi_{z_1,z_2}$ on $((0,1]^2, \mathcal{B}((0,1]^2))$, determined by
\begin{align*}
& \xi_{z_1,z_2}(B \times C)
 := \int_0^1 \int_0^1 \mathbf{1}_B((1-(uz_1^{r/\lambda})^{\lambda/r})v+(uz_1^{r/\lambda})^{\lambda/r}) \\ 
& \qquad \qquad \quad\qquad \qquad \quad\qquad \qquad \quad \cdot \mathbf{1}_C((1-(uz_2^{r/\lambda})^{\lambda/r})v+(uz_2^{r/\lambda})^{\lambda/r}) \dif u \dif v\\
& \qquad  = \int_0^1 \int_0^1 \mathbf{1}_B((1-u^{\lambda/r}z_1)v+u^{\lambda/r}z_1) \mathbf{1}_C((1-u^{\lambda/r}z_2)v+u^{\lambda/r}z_2) \dif u \dif v
,
\end{align*}
is a Markovian coupling of $\widetilde{U}(z_1,\cdot)$ and $\widetilde{U}(z_2,\cdot)$. Moreover, we obtain
\begin{align*}
    \int_{(0,1]^2} \vert x-y \vert \, \xi_{z_1,z_2}(\dif x \dif y)
    = \vert z_1 - z_2 \vert \int_0^1 \int_0^1 (1-v)u^{\lambda/r} \dif v \dif u
    =\frac{r}{2\alpha}.
\end{align*}

\textbf{2nd case:} Assume that $r=\alpha-\lambda<0$, which may be interpreted as light-tailed reference measure scenario.
Consider the bijection $T\colon [0,\infty) \to [1,\infty)$ given by $T(x)=e^{\alpha x/2}$ with $T^{-1}(x)=-2\log(x)/\alpha$. For $(z,B)\in [1,\infty)\times \mathcal{B}([1,\infty))$, we have
\begin{align*}
& \widetilde{U}(z,B)  = \frac{1}{z^{-2r/\alpha}} \int_0^{z^{-2r/\alpha}}  \frac{\lambda}{(\max\{1,t\})^{\lambda/r}} \int_{T^{-1}(B)\cap [-\frac{\log(\max\{1,t\})}{r},\infty)} e^{-\lambda y} \dif y \dif t \\
& = \frac{1}{z^{-2r/\alpha}} \int_0^{z^{-2r/\alpha}}  \frac{\lambda}{(\max\{1,t\})^{\lambda/r}} \int_{T^{-1}(B\cap [\max\{1,t\}^{-\alpha/(2r)},\infty))} e^{-\lambda y} \dif y \dif t \\
&= \frac{1}{ z^{-2r/\alpha}} \int_0^{z^{-2r/\alpha}}  \frac{2\lambda}{(\max\{1,t\})^{\lambda/r}\alpha} \int_{B\cap [\max\{1,t\}^{-\alpha/(2r)},\infty)}  y^{-2\lambda/\alpha-1} \dif y \dif t,
\end{align*}
where in the last equality we applied integration by substitution w.r.t. $T^{-1}$. Given $t>0$, define the inner integral induced distribution 
\[
\widetilde{\nu}_t(B) := \frac{2\lambda}{(\max\{1,t\})^{\lambda/r}\alpha} \int_{B\cap [\max\{1,t\}^{-\alpha/(2r)},\infty)}  y^{-2\lambda/\alpha-1} \dif y, \quad B\in\mathcal{B}([1,\infty)).
\]
With a random variable $U\sim \text{Unif}([0,1])$
we have for all $s \geq (\max\{1,t\})^{-\alpha/(2r)}$ that
\[
\widetilde{\nu}_t((-\infty,s]) = 1-(\max\{1,t\})^{\lambda/r} s^{-2\lambda/\alpha} = \mathbb{P}((\max\{1,t\})^{-\alpha/(2r)} U^{-\alpha/(2\lambda)} \leq s)
\]
and for all $s<(\max\{1,t\})^{-\alpha/(2r)}$ that 
\[
\widetilde{\nu}_t((-\infty,s])=0=\mathbb{P}((\max\{1,t\})^{-\alpha/(2r)} U^{-\alpha/(2\lambda)} \leq s),\]
i.e. the corresponding CDFs coincide and $(\max\{1,t\})^{-\alpha/(2r)} U^{-\alpha/(2\lambda)}\sim\widetilde{\nu}_t$. 

For any $z_1,z_2\in[1,\infty)$ the probability measure $\xi_{z_1,z_2}$ on $([1,\infty)^2, \mathcal{B}([1,\infty)^2))$, induced by satisfying 
\begin{align*}
 \xi_{z_1,z_2}(B \times C)
 & := \int_0^1 \int_0^1 \mathbf{1}_B((\max\{1,v z_1^{-2r/\alpha}\})^{-\alpha/(2r)}u^{-\alpha/(2\lambda)})  \\
 & \qquad \qquad \qquad \cdot \mathbf{1}_C((\max\{1,v z_2^{-2r/\alpha}\})^{-\alpha/(2r)}u^{-\alpha/(2\lambda)}) \dif u \dif v
,
\end{align*}
is a Markovian coupling of $\widetilde{U}(z_1,\cdot)$ and $\widetilde{U}(z_2,\cdot)$. 
W.l.o.g. let $z_1\leq z_2$ and note that
\begin{align*}
    &  \int_{[1,\infty)^2} \vert x-y \vert \, \xi_{z_1,z_2}(\dif x \dif y)\\
    =& \int_0^1 \int_0^1 u^{-\alpha/(2\lambda)} \vert (\max\{1,v z_1^{-2r/\alpha}\})^{-\alpha/(2r)}- (\max\{1,v z_2^{-2r/\alpha}\})^{-\alpha/(2r)}\vert \dif v \dif u\\
    =& \frac{2\lambda}{2\lambda-\alpha}\int_0^1  \vert (\max\{1,v z_1^{-2r/\alpha}\})^{-\alpha/(2r)}- (\max\{1,v z_2^{-2r/\alpha}\})^{-\alpha/(2r)}\vert \dif v
\end{align*}
Recall that $r<0$ and split the remaining integral into different regions to treat the max-function, such that
\begin{align*}
   & \int_0^1  \vert (\max\{1,v z_1^{-2r/\alpha}\})^{-\alpha/(2r)}- (\max\{1,v z_2^{-2r/\alpha}\})^{-\alpha/(2r)}\vert \dif v\\
   = & \int_{z_2^{2r/\alpha}}^{z_1^{2r/\alpha}} (v^{-\alpha/(2r)}z_2-1) \dif v + (z_2-z_1) \int_{z_1^{2r/\alpha}}^1 v^{-\alpha/(2r)} \dif v\\
   = & \frac{2r}{2r-\alpha}(z_2-z_1) - \Big( 1- \frac{2r}{2r-\alpha}\Big) (z_1^{2r/\alpha}-z_2^{2r/\alpha})
   \leq \frac{2r}{2r-\alpha}(z_2-z_1).
\end{align*}
Therefore, we have
\[
\int_{[1,\infty)^2} \vert x-y \vert \, \xi_{z_1,z_2}(\dif x \dif y) \leq \frac{2\lambda}{2\lambda-\alpha} \cdot \frac{2r}{2r-\alpha}\vert z_2-z_1 \vert = \big( 1-\big(\frac{\alpha}{2\lambda-\alpha}\big)^2\big) \vert z_2-z_1 \vert.
\]

Taking the final estimates of both cases together, in the light of \cite{ollivier2009ricci}, we obtain $1-\gamma(\alpha,\lambda)$ as an upper bound for the Ricci curvature. Applying \cite[Proposition~30]{ollivier2009ricci} and Remark~\ref{rem: SPI} yields $\gamma(\alpha,\lambda)$ as lower bound for the right hand-side of \eqref{eq: spec_gap_transf}, which finishes the proof. 
\end{proof}

\section{Degenerating condition profile}\label{sec:degenerating}

Throughout this section, we impose as standing assumption that the potential $V\colon\mathbb{R}^d \to \mathbb{R}$ satisfies $V(x) \geq V(0) \geq 0$ for all $x\in\mathbb{R}^d$ and $\lim_{\vert x \vert \to \infty} V(x) = \infty$.
We start with two auxiliary tools.

\begin{lemma}\label{lem:bulk-conditioning}

For potential $V\colon \mathbb{R}^d \to [0,\infty)$ we assume that there are constants
$c_{1},c_{2},p_{1},p_{2},r_{1},r_{2}\in (0,\infty)$ and $v^{*} \in (0,\infty)$
with 
\begin{align*}
    p_{1}\geq p_{2},\quad r_{1}\geq r_{2},\quad (r_{1}/c_{2})^{\,p_{2}}\geq v_{*},
\end{align*}
such that for all $x\in\mathbb{R}^d$ we have
\begin{align}
\notag
\left\vert x\right\vert  \leq r_{1} & \; \implies\; (\vert x\vert/c_{1})^{p_{1}} < V(x) <(\vert x\vert/c_{2})^{p_{2}},\\
\label{al: r_2}
\vert x\vert >r_{2} & \; \implies\; v_{*} < V(x).
\end{align}
Then, for all $v \in ( 0, v_{*}]$ we have
\begin{align*}
    B(0,c_2 v^{1/p_{2}}) 
    \subseteq\{x\in\mathbb{R}^d\colon V\left(x\right) < v\} \subseteq
    B(0,c_1 v^{1/p_{1}}).
\end{align*}
\end{lemma}
\begin{proof}
Firstly, for $v\in[0,v_{*}]$ assume that $x\in B(0,c_2v^{1/p_{2}}) $, i.e. $\vert x \vert \leq c_2 v^{1/p_2}$. By $v_{*} \leq (r_1/c_{2})^{p_2}$ we have $\vert x\vert \leq r_1$ such that $V(x)<(\vert x \vert/c_{2})^{p_2} \leq v$ and the first inclusion is verified. 
Secondly, for $v\in[0,v_{*}]$ assume that $V(x)<v$, such that by the contrapositive of \eqref{al: r_2} we have $\vert x \vert \leq r_2 \leq r_1$. Hence $(\vert x\vert/c_{1})^{p_{1}} \leq V(x) < v$, which by rearrangement implies that $\vert x\vert  \leq c_1 v^{1/p_{1}}$, yielding the second inclusion.
\end{proof}

\begin{lemma}\label{lem:tail-conditioning}
For potential $V\colon \mathbb{R}^d \to [0,\infty)$ we assume that there are
constants $C_{1},C_{2},q_1,q_2,R_{1},R_{2}\in(0,\infty)$ and $v^*\in(0,\infty)$
with
\begin{align*}
    q_1\leq q_2\quad R_{1} \leq R_{2},\quad  
        (R_{1}/C_{2})^{\,q_2} \leq v^*,
\end{align*}
such that for all $x\in\mathbb{R}^d$ we have
\begin{align}
\label{al: R_1}
\vert x\vert \geq R_{1} & \; \implies\; (\vert x\vert/C_{1})^{q_1} < V(x) < (\vert x\vert/C_{2} )^{q_2},\\
\label{al: R_2}
\notag
\vert x\vert <R_{2} & \; \implies\; V\left(x\right)<v^*.
\end{align}
Then, for $v\in[v^*,\infty)$ we have
\begin{align*}
    B(0,C_2v^{1/q_2}) \subseteq\{x\in\mathbb{R}^d\colon V(x) < v\} 
    \subseteq B(0,C_1v^{1/q_1}).
\end{align*}
\end{lemma}
\begin{proof}
Firstly, for $v\in [v^*,\infty)$ assume that $x\in B(0,C_2v^{1/q_2})$, i.e. $\vert x\vert \leq C_2v^{1/q_2}$. On the one hand, if $\vert x \vert <R_1 \leq R_2$, then $V(x)<v^* \leq v$. 
On the other hand, if $R_1\leq \vert x \vert$, by \eqref{al: R_1} and $\vert x \vert \leq C_2v^{1/q_2}$ follows $V(x)<(\vert x \vert/C_{2})^{q_2}\leq v$. This proves the first inclusion. Secondly, for $v\in [v^*,\infty)$ assume that $V(x)<v$. Then, either $\vert x \vert <R_1$ or $\vert x \vert \geq R_1$. In the latter case, by \eqref{al: R_1}, we have $(\vert x\vert/C_{1}) ^{q_1} < V(x) < v$, such that $\vert x\vert <C_1v^{1/q_1}$. 
Therefore, in general (without the restriction $\vert x \vert \geq R_1$) we obtain
$\vert x\vert  \leq \max\{ R_{1},C_1v^{1/q_1}\} =C_1v^{1/q_1}$. Here the last equality follows by the fact that
$
v\geq v^* 
\geq (R_{1}/C_{2})^{\,q_2}
> ( R_{1}/C_{1})^{\,q_1},
$ where the final inequality is implied by \eqref{al: R_1}.
This concludes the proof of the second inclusion.
\end{proof}

Starting from the results of the latter lemmas we derive estimates that are applicable in examples of Hit-and-Run within slice sampling. Before turning to that
we add an auxiliary estimate that follows by a suitable case distinction.
\begin{lemma}
\label{lem: aux_est_of_fcts}
    Let $c^*,c_*,c \in [0,\infty)$, $\alpha,\beta\in (0,\infty)$ and $v^*,v_* \in (0,\infty)$ with $v_*\leq v^*$. Then, for any $t\in (0,1)$ with
    \[
        f(t) := \begin{cases}
                    c^* (\log t^{-1})^{\beta} & t\in (0,\mathrm{e}^{-v^*}], \\
                    c & t\in(\mathrm{e}^{-v^*},\mathrm{e}^{-v_*}],\\
                    c_* (\log t^{-1})^{-\alpha} & t\in (\mathrm{e}^{-v_*},1),
                \end{cases}
    \]
    we have 
    \[
        f(t) \leq \max\{c^*,c,c_*\}
        \begin{cases}
                    (\log t^{-1})^{\beta} & t\in (0,\mathrm{e}^{-1}], \\
                    (\log t^{-1})^{-\alpha} & t\in (\mathrm{e}^{-1},1).
                \end{cases}
    \]
\end{lemma}

Now we formulate and prove the aforementioned tools for applying our comparison machinery in the context of Hit-and-Run within slice sampling.

\begin{proposition} \label{prop: cond_profile_not_unif_bnd}
    For potential $V\colon \mathbb{R}^d \to [0,\infty)$ assume that there are constants
    $v_*,v^*\in (0,\infty)$ satisfying $v_*\leq v^*$ and
    $c_1,c_2,p_1,p_2,C_1,C_2,q_1,q_2\in(0,\infty)$ with $p_2\leq p_1$, $q_1\leq q_2$ such that for all $v\in (0,v_*)$ holds 
    \begin{align} \label{balls_small_v}
    B(0,c_2 v^{1/p_{2}}) 
    \subseteq\{x\in\mathbb{R}^d\colon V\left(x\right) < v\} \subseteq
    B(0,c_1 v^{1/p_{1}})    
    \end{align}
    and for all $v\in [v^*,\infty)$ holds
    \begin{align} \label{balls_large_v}
    B(0,C_2v^{1/q_2}) \subseteq\{x\in\mathbb{R}^d\colon V(x) < v\} 
    \subseteq B(0,C_1v^{1/q_1}).
    \end{align}
    Define
    \begin{align*}
        b_1 & := 2^{33} d^2\max\left\{\frac{c_1}{c_2},\frac{C_1}{C_2},\frac{C_1}{c_2}\cdot \frac{(v^*)^{1/q_1}}{(v_*)^{1/p_2}}\right\},
        \quad
        b_2:=  \omega_d\, 
    \max\left\{  \frac{(v^*)^{d/q_1}C_1^d}{(v_*)^{d/p_1}}, c_1^d \right\}
    \end{align*}
    with $\omega_d:={\rm Leb}(B(0,1))$ and set $r^* := \min\{q_1,p_1\}$.
    Then, with $\varpi\colon \mathbb{R}^d\to(0,1]$ given by $\varpi(x)=\exp(-V(x))$ (cf. Definition~\ref{def: potential}), for any $t\in(0,1)$ we have
    \begin{equation}
    \label{eq: kappa_bnd}
        \kappa_{\varpi}(t) \leq  
        \frac{b_1}{2^{33}d^2}\cdot
        \begin{cases}
            (\log t^{-1})^{1/q_1-1/q_2} & t\in (0,\mathrm{e}^{-1}], \\
            (\log t^{-1})^{1/p_2-1/p_1} & t\in (\mathrm{e}^{-1},1),
        \end{cases}                             
    \end{equation}
    where $\kappa_\varpi$ has been defined in Corollary~\ref{cor:qlc-hss-handr}, and
    \begin{align}
     & m(t)  = {\rm Leb}(\sfG(t))
           \leq b_2 \cdot 
        \begin{cases}
          (\log t^{-1})^{d/
        r^*} &      t\in(0,\mathrm{e}^{-1}],\\
          (\log t^{-1})^{d/
        r^*} &      t\in(\mathrm{e}^{-1},1).
        \end{cases}  
        \label{al: m(t)_est}
    \end{align}
    In particular, 
    for any $s>0$, this yields 
    \begin{align}
    \notag
           \int_0^1 \mathbf{1}_{(0,2^{33}\, d^{2} \kappa_{\varpi}(t)^{2})}(s) \, m \left( t \right) \dif t
          & \leq b_2 
          \Big(\min\big\{1,(b_1/s)^{\frac{p_1-p_2}{p_1p_2}}\big\}\Big)^{d/r^*+1}\\
          & \qquad +
          2b_2\,\Big(\frac{2d}{r^* \mathrm{e}}\Big)^{d/r^*} \mathrm{e}^{-\frac{1}{2}\max\big\{1,(s/b_1)^{\frac{q_2-q_1}{q_1q_2}}\big\}}.
          \label{al: beta_int_est}
    \end{align}
\end{proposition}
\begin{proof}
    By \eqref{balls_small_v}, \eqref{balls_large_v}, the fact that 
    \[
        \sfG(t) = \{ x\in \mathbb{R}^d\colon V(x)<\log t^{-1} \},
    \]
    and the definition of $\kappa_\varpi$ we obtain
    \[
    \kappa_\varpi(t) \leq 
    \begin{cases}
        \frac{C_1}{C_2} (\log t^{-1})^{1/q_1-1/q_2} & t\in(0,\mathrm{e}^{-v^*}], 
          \\  
        \frac{C_1 (v^*)^{1/q_1}}{c_2(v_*)^{1/p_2}} & t\in(\mathrm{e}^{-v^*},\mathrm{e}^{-v_*}],    \\
        \frac{c_1}{c_2} (\log t^{-1})^{1/p_1-1/p_2} & t\in(\mathrm{e}^{-v_*},1).
    \end{cases}
    \]
    By Lemma~\ref{lem: aux_est_of_fcts} the right-hand side of the former inequality can be further estimated, which implies \eqref{eq: kappa_bnd}.

    We turn to the verification of \eqref{al: m(t)_est}. Again, by leveraging \eqref{balls_small_v} and \eqref{balls_large_v}, we obtain
    \begin{align*}
        m(t) & \leq  \omega_d
    \begin{cases}
        C_1^d (\log t^{-1})^{d/q_1} & t\in(0,\mathrm{e}^{-v^*}], 
          \\  
        C_1^d (v^*)^{d/q_1} & t\in(\mathrm{e}^{-v^*},\mathrm{e}^{-v_*}],    \\
        c_1^d (\log t^{-1})^{d/p_1} & t\in(\mathrm{e}^{-v_*},1),
    \end{cases}\\
    & \leq  \omega_d \max\left\{  \frac{(v^*)^{d/q_1}C_1^d}{(v_*)^{d/p_1}}, c_1^d \right\}
    \begin{cases}
           (\log t^{-1})^{d/q_1} & t\in(0,\mathrm{e}^{-v^*}], 
          \\  
        (\log \mathrm{e}^{v_*})^{d/p_1} & t\in(\mathrm{e}^{-v^*},\mathrm{e}^{-v_*}], \\
        (\log t^{-1})^{d/p_1} & t\in(\mathrm{e}^{-v_*},1),
    \end{cases}
    \\
    & \leq \omega_d \max\left\{ \frac{(v^*)^{d/q_1} C_1^d }{(v_*)^{d/p_1}}, c_1^d \right\} \max\{ (\log t^{-1})^{d/q_1},(\log t^{-1})^{d/p_1} \}.
    \end{align*}
    By representing the last maximum according to the case distinction whether $\log t^{-1}$ is larger or smaller than 1, finishes this part.

    Finally we prove \eqref{al: beta_int_est}. By splitting the integration domain we get
    \begin{align*}
         \int_0^1 \mathbf{1}_{(0,2^{33}\, d^{2} \kappa_{\varpi}(t)^{2})}(s) \, m \left( t \right) \dif t 
         \leq & \int_0^{\mathrm{e^{-1}}} \mathbf{1}_{(0,2^{33}\, d^{2} \kappa_{\varpi}(t)^{2})}(s) \, m \left( t \right) \dif t \\
          & + \int_{\mathrm{e^{-1}}}^1 \mathbf{1}_{(0,2^{33}\, d^{2} \kappa_{\varpi}(t)^{2})}(s) \, m \left( t \right) \dif t.
    \end{align*}
    We consider both integrals separately. By \eqref{eq: kappa_bnd} and \eqref{al: m(t)_est} we obtain for the first integral that
    \begin{align*}
        &  \int_0^{\mathrm{e^{-1}}} \mathbf{1}_{(0,2^{33}\, d^{2} \kappa_{\varpi}(t)^{2})}(s) \, m \left( t \right) \dif t \\
        \leq & b_2 \int_0^{\mathrm{e^{-1}}} \mathbf{1}_{(0,b_1 (\log t^{-1})^{1/q_1-1/q_2})}(s) (\log t^{-1})^{d/r^*} \, \dif t \\
        = & b_2 \int_1^{\infty} \mathbf{1}_{\big((s/b_1)^{\frac{q_2-q_1}{q_1q_2}},\infty\big)}(u)\, u^{d/r^*} \mathrm{e}^{-u} \dif u,
    \end{align*}
    where the last equality follows by substituting $t=\mathrm{e}^{-u}$. By the elementary inequality $u^{d/r^*} \mathrm{e}^{-u} \leq \big(\frac{2d}{r^{*}\mathrm{e}}\big)^{d/r^*} \mathrm{e}^{-u/2}$ for $u\in[1,\infty)$ we obtain
       \begin{align*}
          \int_0^{\mathrm{e^{-1}}} \mathbf{1}_{(0,2^{33}\, d^{2} \kappa_{\varpi}(t)^{2})}(s) \, m \left( t \right) \dif t
       &  \leq b_2 \Big(\frac{2d}{r^{*}\mathrm{e}}\Big)^{d/r^*}  \int_{\max\big\{1,(s/b_1)^{\frac{q_2-q_1}{q_1 q_2}}\big\}}^\infty \mathrm{e}^{-u/2}\, \dif u \\
       & = 2b_2 \Big(\frac{2d}{r^{*}\mathrm{e}}\Big)^{d/r^*} \mathrm{e}^{-\frac{1}{2}\max\big\{1,(s/b_1)^{\frac{q_2-q_1}{q_1q_2}}\big\}}.
        \end{align*}
        By \eqref{eq: kappa_bnd}, \eqref{al: m(t)_est} and also the substitution $t=\mathrm{e}^{-u}$ we obtain for the second integral that
    \begin{align*}
        &  \int_{\mathrm{e^{-1}}}^1 \mathbf{1}_{(0,2^{33}\, d^{2} \kappa_{\varpi}(t)^{2})}(s) \, m \left( t \right) \dif t 
        \leq  b_2 \int_0^{1} \mathbf{1}_{\big(0,(b_1/s)^{\frac{p_1-p_2}{p_1 p_2}} \big)} (u)  u^{d/r^*} \mathrm{e}^{-u}\, \dif u\\
        \leq & b_2 \int_0^{\min\big\{ 1,(b_1/s)^{\frac{p_1-p_2}{p_1 p_2}} \big\}} u^{d/r^*} \,\dif u \leq b_2 \Big (\min\big\{ 1,(b_1/s)^{\frac{p_1-p_2}{p_1 p_2}} \big\} \Big)^{d/r^*+1}.
    \end{align*}
    Taking the estimates of both integrals together finishes the proof of \eqref{al: beta_int_est}.
\end{proof}

In particular, if $p_1=p_2$ (i.e. the smallest level sets are well-conditioned, so that $\kappa_\varpi$ is uniformly bounded), then the right hand-side of \eqref{al: beta_int_est} decays at a rate which is at least exponential in $s$, since the first term of the upper bound vanishes. By contrast, if $ p_2<p_1$, then the first term dominates, and as a consequence, the right hand-side of \eqref{al: beta_int_est} decays only polynomial in $s$. 

\begin{example} \label{ex: appendix}
Fix $I\subseteq \{1,\dots,d\}$ 
and let $I^c:=\{1,\dots,d\}\setminus I$. We consider $V(x) = \sum_{i\in I} x_i^2 + \sum_{i\in I^c} x_i^4$ for $x=(x_1,\dots,x_d)\in\mathbb{R}^d$. For all $i\in \{1,\dots,d\}$ and $x\in\mathbb{R}^d\setminus\{0\}$ we have 
\begin{equation} 
\label{eq: smaller_1}
    \vert x_i \vert  \leq 
    \vert x \vert,\quad \text{such that} \quad \frac{\vert x_i \vert^2}{\vert x\vert^2} \geq \frac{\vert x_i \vert^4}{\vert x\vert^4},
\end{equation}
which implies
\begin{equation}
    \label{eq: upper}
    V(x) \underset{\eqref{eq: smaller_1}} {\leq} \vert x \vert^2 \sum_{i\in I} \frac{x_i^2}{\vert x \vert^2} + \vert x \vert^4 \sum_{i\in I^c} \frac{x_i^2}{\vert x \vert^2}
\end{equation}
and
\begin{equation}
    \label{eq: lower}
    V(x) \underset{\eqref{eq: smaller_1}}{\geq} \vert x \vert^2 \sum_{i\in I} \frac{x_i^4}{\vert x \vert^4} + \vert x \vert^4 \sum_{i\in I^c} \frac{x_i^4}{\vert x \vert^4}.
\end{equation}
Therefore, for $x\in B(0,1)\setminus\{0\}$, which yields $\vert x \vert^2 \geq \vert x \vert^4$, we obtain
$
    V(x) \underset{\eqref{eq: upper}}{\leq} \vert x \vert^2    
$
and 
\[
 V(x) \underset{\eqref{eq: lower}}{\geq} \vert x \vert^4 \sum_{i=1}^d \frac{x_i^4}{\vert x\vert^4} \geq \frac{\vert x \vert ^4}{d},
\]
where the last inequality follows by
          $\sum_{i=1}^d y_i^4 \geq \frac{\vert y\vert^2}{d}$,
with $y=x/\vert x \vert \in \mathbb{R}^d$.
Consequently, for all $x\in B(0,1)$, we have
\begin{equation}  \label{eq: V_small}
    \vert x \vert^4/d \leq V(x) \leq \vert x \vert^2.
\end{equation}
On the other hand, for $x\not\in B(0,1)$, which implies $\vert x\vert^2 \leq \vert x\vert^4$, we have $V(x) \underset{\eqref{eq: upper}}{\leq} \vert x \vert^4$ and
\[
V(x)\underset{\eqref{eq: lower}}{\geq} \vert x \vert^2 \sum_{i=1}^d \frac{x_i^4}{\vert x\vert^4} \geq \frac{\vert x \vert ^2}{d},
\]
where (again) the last inequality follows by
          $\sum_{i=1}^d y_i^4 \geq \frac{\vert y\vert^2}{d}$,
with $y=x/\vert x \vert \in \mathbb{R}^d$.
Consequently, for all $x\not\in B(0,1)$, we obtain
\begin{equation} \label{eq: V_large}
    \vert x \vert^2/d \leq V(x) \leq \vert x \vert^4.
\end{equation}
Now we are able to apply Lemma~\ref{lem:bulk-conditioning} with $p_1=4$, $p_2=2$, $c_1=d^{1/4}$, $c_2=1$, $r_1=r_2=1$ (by using \eqref{eq: V_small}) and $v_*=1/d$ (by using \eqref{eq: V_large}) such that for all $v\in (0,1/d]$ holds
\[
    B(0,v^{1/2}) \subseteq \{x\in\mathbb{R}^d\colon V(x) < v\} \subseteq B(0,(dv)^{1/4}).    
\]
Moreover, we apply Lemma~\ref{lem:tail-conditioning} with $q_1=2$, $q_2=4$, $C_1=d^{1/2}$, $C_2=1$, $R_1=R_2=1$ (by using \eqref{eq: V_small}) and $v^*=1$ (by using \eqref{eq: V_large})
such that for all $v\in [1,\infty)$ holds
\[
    B(0,v^{1/4}) \subseteq \{x\in\mathbb{R}^d\colon V(x) < v\} \subseteq B(0,(dv)^{1/2}).    
\]
Consequently, the assumptions of Proposition~\ref{prop: cond_profile_not_unif_bnd} are satisfied with the previously set parameters, such that
\[
r^*=2,\qquad  b_1 = 2^{33} d^3,\qquad b_2 = \omega_d d^{3/4d},
\]
where $\omega_d={\rm Leb}(B(0,1))$, and for any $s>0$, we have
\begin{align*}
\int_0^1 \mathbf{1}_{(0,2^{33}\, d^{2} \kappa_{\varpi}(t)^{2})}(s) \, m \left( t \right) \dif t
          & \leq \omega_d d^{3/4d}
          \Big(\min\big\{1,(2^{33} d^3/s)^{1/4}\big\}\Big)^{d/2+1}\\
          & \qquad 
          +
          2\omega_d d^{3/4d}\,\Big(\frac{d}{\mathrm{e}}\Big)^{d/2} \mathrm{e}^{-\frac{1}{2}\max\big\{1,(s/(2^{33}d^3))^{1/4}\big\}}.
\end{align*}
By the fact that $\mathrm{e}^{-x/2} \leq (d/2+1)  x^{-d/2-1}$ for all $x\geq 1$ we further bound 
\begin{align*}
& \int_0^1 \mathbf{1}_{(0,2^{33}\, d^{2} \kappa_{\varpi}(t)^{2})}(s) \, m \left( t \right) \dif t \\
& \leq  \omega_d(1+2(d/2+1) (d/\mathrm{e})^{d/2} )        
      d^{3/4d}\Big(\min\big\{1,(2^{33} d^3/s)^{1/4}\big\}\Big)^{d/2+1}\\
& \leq \omega_d(1+(d+2)(d/\mathrm{e})^{d/2})     d^{3/4d} (2^{33} d^3/s)^{d/8+1/4}.
\end{align*}
Finally, a few more (crude) estimates for $d\geq4$ yield
\begin{equation}
\label{eq: main_ex_app}
    \int_0^1 \mathbf{1}_{(0,2^{33}\, d^{2} \kappa_{\varpi}(t)^{2})}(s) \, m \left( t \right) \dif t
    \leq \omega_d\, d^{2d}\, 2^{6d}\,  s^{-d/8-1/4}.
\end{equation}

\end{example}

\bibliography{slice_comp}
\bibliographystyle{plainnat}

\end{document}

%% file: macros.tex
\newtheorem{theorem}{Theorem}
\newtheorem{lemma}[theorem]{Lemma}
\newtheorem{proposition}[theorem]{Proposition}
\newtheorem{assumption}{Assumption}
\newtheorem{corollary}[theorem]{Corollary}
\newtheorem{remark}[theorem]{Remark}
\newtheorem{definition}[theorem]{Definition}
\newtheorem{example}[theorem]{\protect\examplename}

\providecommand{\examplename}{Example}

\newcommand{\R}{\mathbb R}
\newcommand{\Rd}{\mathbb R^d}

\newcommand{\Id}{\mathrm{Id}}
\newcommand{\ELL}{\mathrm{L}^2}
\newcommand{\calE}{\mathcal E}

\newcommand{\sfG}{\mathsf G}
\newcommand{\osc}{\mathrm{osc}}
\newcommand{\sfT}{\mathsf T}
\newcommand{\Var}{\mathrm{Var}}

\newcommand{\ess}{\mathrm{ess}}

%% file: paper_revision.bbl
\begin{thebibliography}{38}
\providecommand{\natexlab}[1]{#1}
\providecommand{\url}[1]{\texttt{#1}}
\expandafter\ifx\csname urlstyle\endcsname\relax
  \providecommand{\doi}[1]{doi: #1}\else
  \providecommand{\doi}{doi: \begingroup \urlstyle{rm}\Url}\fi

\bibitem[Andrieu et~al.(2018)Andrieu, Lee, and Vihola]{Andrieu2018uniform}
Christophe Andrieu, Anthony Lee, and Matti Vihola.
\newblock {Uniform ergodicity of the iterated conditional SMC and geometric
  ergodicity of particle Gibbs samplers}.
\newblock \emph{Bernoulli}, 24\penalty0 (2):\penalty0 842--872, 2018.
\newblock \doi{10.3150/15-BEJ785}.

\bibitem[Andrieu et~al.(2022{\natexlab{a}})Andrieu, Lee, Power, and
  Wang]{Andrieu2021comparison}
Christophe Andrieu, Anthony Lee, Sam Power, and Andi~Q. Wang.
\newblock {Comparison of Markov chains via weak Poincar\'e inequalities with
  application to pseudo-marginal MCMC}.
\newblock \emph{The Annals of Statistics}, 50\penalty0 (6):\penalty0
  3592--3618, 2022{\natexlab{a}}.

\bibitem[Andrieu et~al.(2022{\natexlab{b}})Andrieu, Lee, Power, and
  Wang]{Andrieu2022poincare}
Christophe Andrieu, Anthony Lee, Sam Power, and Andi~Q. Wang.
\newblock {Poincar\'e inequalities for Markov chains: a meeting with Cheeger,
  Lyapunov and Metropolis}.
\newblock Technical report, {University of Bristol}, 2022{\natexlab{b}}.
\newblock URL \url{https://arxiv.org/abs/2208.05239v1}.

\bibitem[Andrieu et~al.(2024)Andrieu, Lee, Power, and Wang]{Andrieu2022rwm}
Christophe Andrieu, Anthony Lee, Sam Power, and Andi~Q. Wang.
\newblock {Explicit convergence bounds for Metropolis Markov chains:
  Isoperimetry, spectral gaps and profiles}.
\newblock \emph{The Annals of Applied Probability}, 34\penalty0 (4):\penalty0
  4022--4071, 2024.

\bibitem[Andrieu et~al.(2025+)Andrieu, Lee, Power, and Wang]{Andrieu2023wpi}
Christophe Andrieu, Anthony Lee, Sam Power, and Andi~Q. Wang.
\newblock {Weak Poincar\'e Inequalities for Markov chains: theory and
  applications}.
\newblock \emph{The Annals of Applied Probability}, 2025+.
\newblock URL \url{https://arxiv.org/abs/2312.11689v1}.
\newblock To appear.

\bibitem[Baxendale(2005)]{Baxendale2005geometric}
Peter~H. Baxendale.
\newblock {Renewal theory and computable convergence rates for geometrically
  ergodic Markov chains}.
\newblock \emph{The Annals of Applied Probability}, 15\penalty0 (1B):\penalty0
  700--738, 2005.
\newblock ISSN 1050-5164.
\newblock \doi{10.1214/105051604000000710}.

\bibitem[Chen et~al.(2018)Chen, Dwivedi, Wainwright, and Yu]{chen2018fast}
Yuansi Chen, Raaz Dwivedi, Martin~J Wainwright, and Bin Yu.
\newblock Fast {MCMC} sampling algorithms on polytopes.
\newblock \emph{The Journal of Machine Learning Research}, 19\penalty0
  (1):\penalty0 2146--2231, 2018.

\bibitem[Chewi(2023)]{chewi2023log}
Sinho Chewi.
\newblock Log-concave sampling.
\newblock \emph{Book draft available at https://chewisinho. github. io}, 2023.

\bibitem[Cotter et~al.(2013)Cotter, Roberts, Stuart, and White]{cotter2013mcmc}
Simon~L Cotter, Gareth~O. Roberts, Andrew~M. Stuart, and David White.
\newblock {MCMC methods for functions: modifying old algorithms to make them
  faster}.
\newblock \emph{{Statistical Science}}, 28\penalty0 (3):\penalty0 424--446,
  2013.

\bibitem[Doucet et~al.(2015)Doucet, Pitt, Deligiannidis, and
  Kohn]{Doucet2015intractable}
Arnaud Doucet, Michael~K. Pitt, George Deligiannidis, and Robert Kohn.
\newblock {Efficient implementation of Markov chain Monte Carlo when using an
  unbiased likelihood estimator}.
\newblock \emph{Biometrika}, 102\penalty0 (2):\penalty0 295--313, 2015.
\newblock ISSN 14643510.
\newblock \doi{10.1093/biomet/asu075}.

\bibitem[Edwards and Sokal(1988)]{Edwards1988}
Robert~G. Edwards and Alan~D. Sokal.
\newblock {Generalization of the Fortuin-Kasteleyn-Swendsen-Wang representation
  and Monte Carlo algorithm}.
\newblock \emph{Physical Review D}, 38\penalty0 (6):\penalty0 2009, sep 1988.
\newblock ISSN 05562821.
\newblock \doi{10.1103/PhysRevD.38.2009}.

\bibitem[G{\aa}semyr(2006)]{Gasemyr2006spectrum}
J{\o}rund G{\aa}semyr.
\newblock {The Spectrum of the Independent Metropolis–Hastings Algorithm}.
\newblock \emph{Journal of Theoretical Probability}, 19\penalty0 (1):\penalty0
  152--165, apr 2006.
\newblock ISSN 1572-9230.
\newblock \doi{10.1007/S10959-006-0009-2}.

\bibitem[Higdon(1998)]{higdon1998auxiliary}
David~M Higdon.
\newblock Auxiliary variable methods for markov chain monte carlo with
  applications.
\newblock \emph{Journal of the American statistical Association}, 93\penalty0
  (442):\penalty0 585--595, 1998.

\bibitem[Jarner and Hansen(2000)]{Jarner2000}
S{\o}ren~Fiig Jarner and Ernst Hansen.
\newblock {Geometric ergodicity of Metropolis algorithms}.
\newblock \emph{Stochastic Processes and their Applications}, 85\penalty0
  (2):\penalty0 341--361, 2000.
\newblock ISSN 0304-4149.
\newblock \doi{10.1016/S0304-4149(99)00082-4}.

\bibitem[John(2014)]{john2014extremum}
Fritz John.
\newblock Extremum problems with inequalities as subsidiary conditions.
\newblock \emph{Traces and emergence of nonlinear programming}, pages 197--215,
  2014.

\bibitem[Kannan and Narayanan(2009)]{kannan2009random}
Ravi Kannan and Hariharan Narayanan.
\newblock Random walks on polytopes and an affine interior point method for
  linear programming.
\newblock In \emph{Proceedings of the forty-first annual ACM symposium on
  Theory of computing}, pages 561--570, 2009.

\bibitem[Laddha and Vempala(2023)]{laddha2023convergence}
Aditi Laddha and Santosh~S Vempala.
\newblock Convergence of {G}ibbs sampling: {C}oordinate {H}it-and-{R}un mixes
  fast.
\newblock \emph{Discrete \& Computational Geometry}, pages 1--20, 2023.

\bibitem[{L}atuszyski and Rudolf(2024)]{Latuszynski2014convergence}
Krzysztof {L}atuszyski and Daniel Rudolf.
\newblock {Convergence of hybrid slice sampling via spectral gap}.
\newblock \emph{Advances in Applied Probability}, 56\penalty0 (4):\penalty0
  1440--1466, 2024.
\newblock \doi{10.1017/APR.2024.16}.

\bibitem[Lawler and Sokal(1988)]{lawler1988bounds}
Gregory~F Lawler and Alan~D Sokal.
\newblock Bounds on the $l^2$ spectrum for {M}arkov chains and {M}arkov
  processes: a generalization of {C}heeger’s inequality.
\newblock \emph{Transactions of the American mathematical society},
  309\penalty0 (2):\penalty0 557--580, 1988.

\bibitem[Lee and Vempala(2022)]{lee2022geodesic}
Yin~Tat Lee and Santosh Vempala.
\newblock Geodesic walks in polytopes.
\newblock \emph{SIAM Journal on Computing}, 51\penalty0 (2):\penalty0
  STOC17--400, 2022.

\bibitem[Lov{\'{a}}sz and Vempala(2004)]{Lovasz2004hitandrun}
L{\'{a}}szl{\'{o}} Lov{\'{a}}sz and Santosh Vempala.
\newblock {Hit-and-run from a corner}.
\newblock \emph{Conference Proceedings of the Annual ACM Symposium on Theory of
  Computing}, pages 310--314, 2004.
\newblock ISSN 07349025.
\newblock \doi{10.1145/1007352.1007403}.

\bibitem[Mira and Tierney(2002)]{Mira2002}
Antonietta Mira and Luke Tierney.
\newblock {Efficiency and Convergence Properties of Slice Samplers}.
\newblock \emph{Scandinavian Journal of Statistics}, 29\penalty0 (1):\penalty0
  1--12, 2002.

\bibitem[Murray et~al.(2010)Murray, Adams, and Mackay]{Murray2010elliptical}
Iain Murray, Ryan~Prescott Adams, and David J~C Mackay.
\newblock {Elliptical slice sampling}.
\newblock In \emph{Proceedings of the 13th International Conference on
  Artificial Intelligence and Statistics}, 2010.

\bibitem[Natarovskii et~al.(2021)Natarovskii, Rudolf, and
  Sprungk]{Natarovskii2021quantitative}
Viacheslav Natarovskii, Daniel Rudolf, and Bj{\"{o}}rn Sprungk.
\newblock {Quantitative spectral gap estimate and Wasserstein contraction of
  simple slice sampling}.
\newblock \emph{Ann. Appl. Probab.}, 31\penalty0 (2):\penalty0 806--825, 2021.
\newblock ISSN 1050-5164.
\newblock \doi{10.1214/20-AAP1605}.

\bibitem[Neal(2003)]{Neal2003Slice}
Radford~M Neal.
\newblock {Slice sampling}.
\newblock \emph{The Annals of Statistics}, 31\penalty0 (3):\penalty0 705--767,
  2003.

\bibitem[Ollivier(2009)]{ollivier2009ricci}
Yann Ollivier.
\newblock {Ricci curvature of Markov chains on metric spaces}.
\newblock \emph{Journal of Functional Analysis}, 256\penalty0 (3):\penalty0
  810--864, 2009.

\bibitem[Qin et~al.(2025+)Qin, Ju, and Wang]{Qin2023}
Qian Qin, Nianqiao Ju, and Guanyang Wang.
\newblock {Spectral gap bounds for reversible hybrid Gibbs chains}.
\newblock \emph{The Annals of Statistics}, 2025+.
\newblock URL \url{https://arxiv.org/abs/2312.12782v1}.
\newblock To appear.

\bibitem[Roberts and Rosenthal(1997)]{Roberts1997geometric}
Gareth Roberts and Jeffrey Rosenthal.
\newblock {Geometric ergodicity and hybrid Markov chains}.
\newblock \emph{Electronic Communications in Probability}, 2:\penalty0 13--25,
  1997.

\bibitem[Roberts and Rosenthal(2002)]{Roberts2002polar}
Gareth Roberts and Jeffrey Rosenthal.
\newblock {The polar slice sampler}.
\newblock \emph{Stochastic Models}, 18\penalty0 (2):\penalty0 257--280, 2002.
\newblock ISSN 15326349.
\newblock \doi{10.1081/STM-120004467}.

\bibitem[Rudolf(2012)]{Ru12}
D.~Rudolf.
\newblock Explicit error bounds for {M}arkov chain {M}onte {C}arlo.
\newblock \emph{Dissertationes Math.}, 485:\penalty0 93 pp., 2012.

\bibitem[Rudolf and Sch{\"a}r(2024)]{RuSch24}
Daniel Rudolf and Philip Sch{\"a}r.
\newblock {Dimension-independent spectral gap of polar slice sampling}.
\newblock \emph{Stat. Comput.}, 20\penalty0 (34), 2024.

\bibitem[Rudolf and Sprungk(2018)]{rudolf2018generalization}
Daniel Rudolf and Bj{\"o}rn Sprungk.
\newblock On a generalization of the preconditioned crank--nicolson metropolis
  algorithm.
\newblock \emph{Foundations of Computational Mathematics}, 18:\penalty0
  309--343, 2018.

\bibitem[Rudolf and Sprungk(2022)]{RudolfSprungk2022}
Daniel Rudolf and Bj\"orn Sprungk.
\newblock {Robust random walk-like Metropolis-Hastings algorithms for
  concentrating posteriors}.
\newblock Technical report, 2022.
\newblock URL \url{https://arxiv.org/abs/2202.12127}.

\bibitem[Rudolf and Ullrich(2013)]{Rudolf2013positivity}
Daniel Rudolf and Mario Ullrich.
\newblock {Positivity of hit-and-run and related algorithms}.
\newblock \emph{Electron. Commun. Probab}, 18\penalty0 (49):\penalty0 1--8,
  2013.
\newblock ISSN 1083589X.
\newblock \doi{10.1214/ECP.v18-2507}.

\bibitem[Rudolf and Ullrich(2018)]{Rudolf2018comparison}
Daniel Rudolf and Mario Ullrich.
\newblock {Comparison of hit-and-run, slice sampler and random walk
  Metropolis}.
\newblock \emph{Journal of Applied Probability}, 55\penalty0 (4):\penalty0
  1186--1202, 2018.
\newblock ISSN 0021-9002.
\newblock \doi{10.1017/JPR.2018.78}.

\bibitem[Sch{\"{a}}r(2023)]{Schar2023revisited}
Philip Sch{\"{a}}r.
\newblock {Wasserstein contraction and spectral gap of slice sampling
  revisited}.
\newblock \emph{The Electronic Journal of Probability}, 28:\penalty0 1--28,
  2023.

\bibitem[Sch{\"{a}}r et~al.(2023)Sch{\"{a}}r, Habeck, and
  Rudolf]{ScharEtAl2023Gibbsian}
Philip Sch{\"{a}}r, Michael Habeck, and Daniel Rudolf.
\newblock {Gibbsian Polar Slice Sampling}.
\newblock In \emph{Proceedings of the 40th International Conference on Machine
  Learning}, volume 202, pages 30204--30223. PMLR, 2023.

\bibitem[Smith(1984)]{Smith1984efficient}
Robert~L. Smith.
\newblock {Efficient Monte Carlo Procedures for Generating Points Uniformly
  Distributed over Bounded Regions}.
\newblock \emph{Operations Research}, 32\penalty0 (6):\penalty0 1296--1308,
  1984.
\newblock ISSN 0030364X.
\newblock \doi{10.1287/OPRE.32.6.1296}.

\end{thebibliography}
